\documentclass{amsart}

\usepackage{fullpage,parskip}
\usepackage{graphicx}
\usepackage{amsmath, amssymb, amsthm}
\usepackage{url}
\usepackage{slashbox}

\usepackage{subcaption}

\newcommand*\flatfour{\includegraphics[width = 1 em]{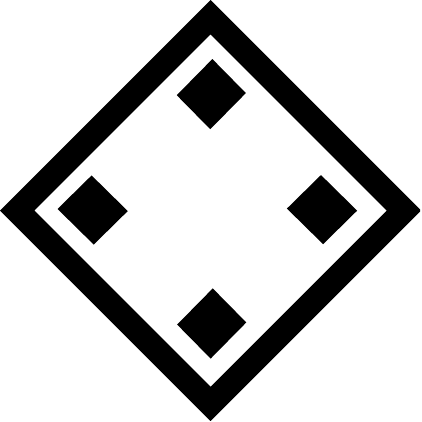}}
\newcommand*\flatadja{\includegraphics[width = 1 em]{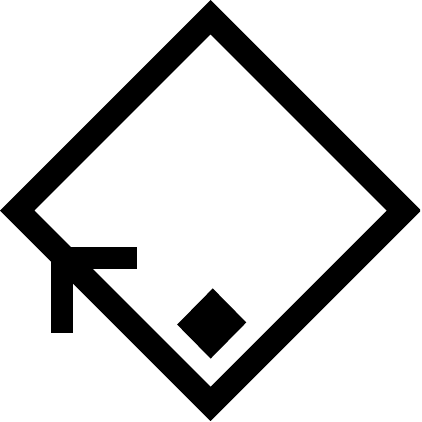}}
\newcommand*\flatadjb{\includegraphics[width = 1 em]{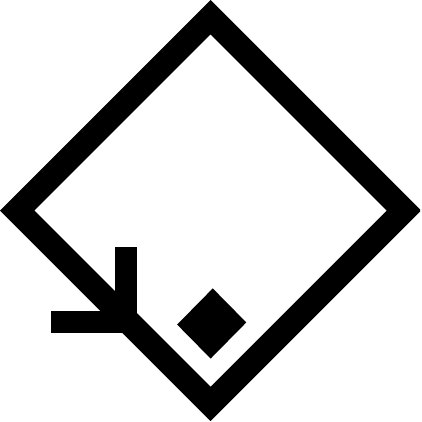}}
\newcommand*\flatoppa{\includegraphics[width = 1 em]{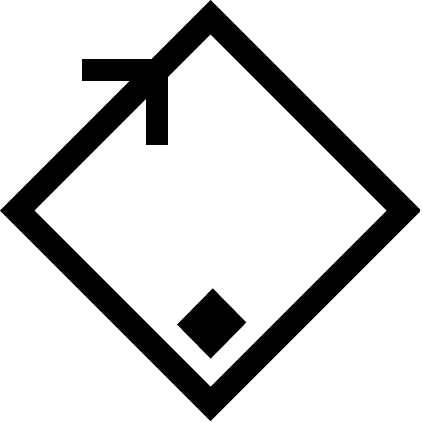}}
\newcommand*\flatoppb{\includegraphics[width = 1 em]{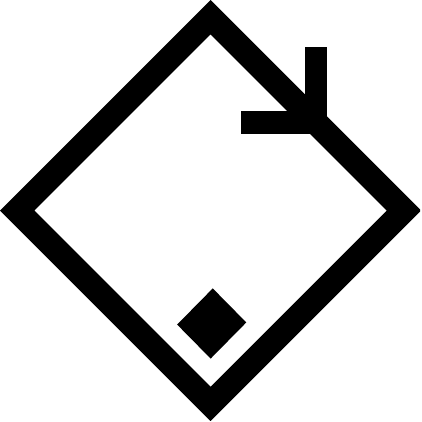}}
\DeclareRobustCommand*\flatadj{\includegraphics[width = 1 em]{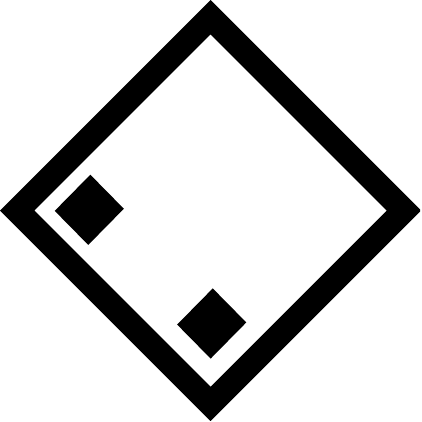}}
\newcommand*\flatopp{\includegraphics[width = 1 em]{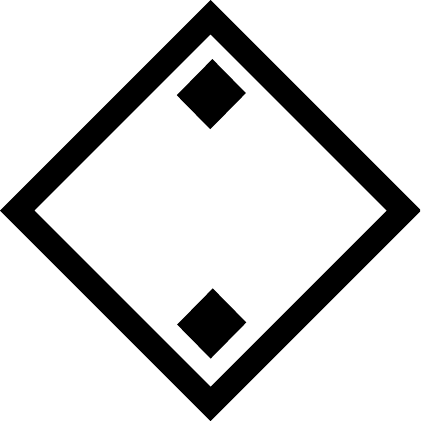}}
\newcommand*\stable{\includegraphics[width = 1 em]{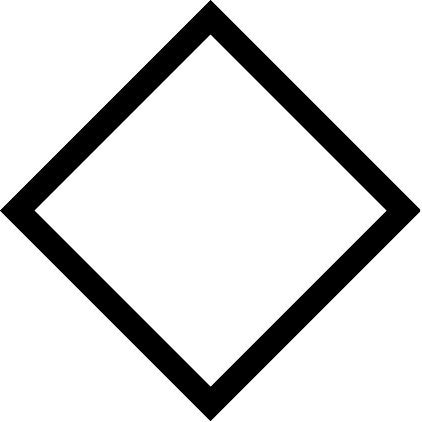}}

\newtheorem{theorem}{Theorem}
\newtheorem{corollary}[theorem]{Corollary}
\newtheorem{proposition}[theorem]{Proposition}
\newtheorem{lemma}[theorem]{Lemma}

\title{The Spread of Cooperative Strategies on  Grids with Random Asynchronous Updating}
\author{Christopher Duffy$^1$\and Jeannette Janssen$^2$\\  \\ $^1$Department of Mathematics and Statistics University of Saskatchewan \\ $^2$Department of Mathematics and Statistics Dalhousie University}
\begin{document}

\begin{abstract}
	The Prisoner's Dilemma Process on a graph $G$ is an iterative process where each vertex, with a fixed strategy (\emph{cooperate} or \emph{defect}), plays the game with each of its neighbours. At the end of a round each vertex  may change its strategy to that of its neighbour with the highest pay-off. Here we study the spread of cooperative and selfish behaviours on a toroidal grid, where each vertex is initially  a cooperator with probability $p$. When vertices are permitted to change their strategies via a randomized asynchronous update scheme, we find that for some values of $p$ the limiting density of cooperators may be modelled as a polynomial in $p$. Theoretical bounds for this density are confirmed via simulation.  
\end{abstract}
	
	\maketitle	
	
\section{Introduction and Preliminaries}

The particular topology of a  network has a dramatic impact on discrete processes that model  competitive interactions in communities \cite{N11}. For example, spread of a particular attitude or belief is less likely to propagate completely in Erd\"{o}s-Renyi graphs than on small-world networks \cite{C10}. Studies of Cellular Automata indicate that the particular updating scheme impacts the limiting configuration of randomly seeded cellular automaton \cite{S99a}. Here we combine these two paradigms to study a discrete-time process that may be modelled as a cellular automata with a particular updating scheme.

The Prisoner's Dilemma, a staple of classical game theory, is a $2$-player game in which each of the two players simultaneously make a decision to either \emph{cooperate} or \emph{defect}. Each of the players receives a pay-off whose amount takes into account the decisions of both players. 
The pay-off structure is chosen so that each player's pay-off is maximised when both choose cooperate and so that a player's pay-off is minimized when they choose cooperate and the other player chooses to defect.  
Classically the game is played in a single round. However by considering the game as being played in a series of rounds, the Prisoner's Dilemma may be used to model a variety of scenarios in many disciplines, including evolutionary biology \cite{S99,D11}, economics \cite{S92,G92} and sociology \cite{H98,S94}. 
More broadly, the Prisoner's Dilemma on graphs fits in the context of evolutionary games on graphs. A  survey of methods and research in this area is given in \cite{S07}.

Here we consider the iterated Prisoner's Dilemma as a game played between neighbours on a graph. In each round each vertex plays, with a fixed strategy  (\emph{cooperate} or \emph{defect}), the game with each of its neighbours. The score for each vertex is the sum of the pay-offs its receives in each game. At the end of each round, each vertex is given the opportunity to update their strategy to that of their most successful neighbour. 

The Prisoner's Dilemma on the grid was first examined by Nowak and May \cite{N93}. Through simulation they find that spatial effects have an impact on the evolution of the strategies of players in the process. In examining symmetric configurations they find ``dynamical fractals" and ``evolutionary kaleidoscopes". In this early work, the authors consider an update scheme in which each of the vertices update simultaneously to emulate the strategy of their most successful neighbour. This updating scheme has been examined on a variety of different graphs, in particular regular lattices, random graphs and small-world networks. For example, Abramson and Kuperman observe a number of properties of a Prisoner's Dilemma process for regular lattices and random graphs \cite{A01};  Dur\'an and Mulet  observe a relationship between initial and final density of cooperators on random graphs when the graph has small connectivity \cite{D05};  Santos et al. consider how the preferential model of attachment in random graphs provide sufficient conditions for the cooperator strategy to propagate \cite{Sa06}.

In much of the previous work in this area, the authors consider a deterministic model of updating in which all of the vertices simultaneously emulate the strategy of their most successful neighbour. However, other update schemes are possible. In \cite{S08}  the authors use a probabilistic process as part of the updating strategies. In addition to  synchronous updating schemes, i.e., those in which all vertices update simultaneously, asynchronous schemes may be studied. In their survey of evolutionary games on graphs~\cite{S07},  Szab\'o and F\'ath consider the update model in which a pair of neighbouring vertices are chosen at random, with one emulating the strategy of the other according to some random variable.

A study of asynchronous update schemes of cellular automata by Cornforth et al. shows that the particular variant of asynchronous update has a dramatic effect on the limiting behaviour of  one dimensional cellular automata \cite{C05}. 
They further highlight the differences between deterministic and probabilistic asynchronous update schemes. 
A full survey of  asynchronous update schemes in cellular automata is given in \cite{C02}. 
The model of asynchronous probabilistic updating presented here provides a new direction in both the study of cellular automata and evolutionary games on graphs. 
We propose a variation of the random independent model of updating that considers the set of vertices envious of their neighbours and updates them in a random order, playing a round of the game after each individual vertex has updated. 
This update scheme behaves similarly to the random independent model for updating cellular automata \cite{S99a}; however, it provides necessary structure to facilitate proofs of observed behaviours. 

We note that in many of these previous works, much of the work has been strictly experimental. That is, emergent behaviours are observed through carefully designed computer simulation. A notable exception to this is the work of  Schweitzer et al. \cite{S02}, whose analysis allows for a verification of the simulations originally presented by Nowak and May. Here we break from this trend to provide theoretical justifications for observed behaviours.
 
In this paper we study the resulting behaviour of the Prisoner's Dilemma process on toroidal grids where  each vertex of the grid  cooperates with probability $p \in [0,1]$.  
Figure \ref{fig:pexample} gives examples of starting and the resulting stable configurations for various values of $p$. 
Here we notice that, though the initial configuration is randomised, the resulting stable configuration exhibits a surprising amount of structure. 
While the update process introduces uncertainty through the choice of the permutation of the envious vertices, we find that for some values of $p$, we may predict the density of cooperators as  $t \to \infty$. 
In Section \ref{sec:evolve} we consider the growth of small clusters existing in infinite grids. 
We use the results from Section \ref{sec:evolve} in Section \ref{sec:Growth} to derive probabilistic bounds on the final density of cooperators on the $n\times n$ toroidal grid for a fixed value of $p$, and for $p$ as a function of $n$. 
 
\section{Preliminaries} 
 Let $G = (V,E)$ be a graph. A \emph{configuration}, $C$, is a function that assigns a strategy to each vertex of $G$. Formally, $C: V \rightarrow \{0,1\}$, where $0$ corresponds to \emph{defector} and $1$ corresponds to \emph{cooperator}. The \emph{pay-off function}, $f$, assigns the score for the first player to an ordered pair that represents the strategies of the first and second player. The pay-off function $f: \{0,1\} \times \{0,1\} \rightarrow \{0,1,T\}$ is given by

\begin{center}
	\begin{tabular}{|c|cc|}
		\hline
	\backslashbox{$C(v_1)$}{$C(v_2)$}	& 0 & 1\\
		\hline $0$ & $0$ & $T$ \\
		$1$ & $0$ & $1$ \\
	\hline
	\end{tabular}
	
\end{center}

%\begin{center}
%	\begin{tabular}{c|c}
%		$C(v_1),C(v_2)$ & $f(C(v_1),C(v_2))$ \\ 
%		\hline $(0,0)$ & $0$ \\ 
%		\hline $(1,0)$ & $0$ \\ 
%		\hline $(0,1)$ & $T$ \\ 
%		\hline $(1,1)$ & $1,$ \\ 
%	\end{tabular}  
%\end{center}

where $T>1$ is a fixed constant. We refer to $T$ as the \emph{cheating advantage}.

Let $C$ be fixed and let $v \in V(G)$. The \emph{score of $v$ with respect to $C$} is given by $$s(v) = \sum_{x \in N(v)} f(C(v),C(x)).$$ When the context is clear we refer to the \emph{score of $v$}. The \emph{most successful neighbours of $v$} are the vertices in the closed neighbourhood of $v$ (denoted $N[v]$) that have the greatest score. We restrict our consideration for possibilities for the value $T$  in such a way that each of the most successful neighbours have the same strategy. Let $u$ be a most successful neighbour of $v$. The vertex $v$ is called \emph{weak with respect to $C$} if $C(v) \neq C(u)$.
Otherwise, we say that $v$ is \emph{strong}. In other words, weak vertices would like to change their strategy and strong vertices are satisfied with their strategy. We are interested in the change in the configuration  with respect to time; we use $C_t$ to denote the configuration at time $t$ and $s_t$ to denote the score at time $t$.

The configuration, $D$, resulting from \emph{updating $v$} with respect to $C$ changes the strategy of $v$ if $v$ is weak and leaves the strategies of all other vertices fixed.

%$$ D(x) = \begin{cases} 
%C(x), & x\neq v, \text{ or}\\
%C(x), & \text{$x$ is strong, or} \\
%1- C(x), & \text{$x$ is weak.}
%\end{cases}
%$$
%Note that when we update $v$ with respect to $C$ it is only the strategy of $v$ and the scores in its closed neighbourhood that possibly change.

We call a maximal connected proper subgraph $K$ of $k$ vertices with the same strategy a \emph{$k$-cluster}.  
We say that a $1$-cluster is an \emph{isolated} vertex. 
The $k$-cluster $H \leq G$ has a \emph{border of width $b$} if for all $h \in V(H)$ and all $v \in V(G - H)$ such that $C(h) = C(v)$, we have that $d(h,v) \geq b$. 
That is, $H$ is surrounded by a border that $b$ vertices wide that consists of vertices all with the opposite strategy.

Given a graph $G$ and some $T>1$, our process is initialised with $C_0$, some configuration of the vertices. 
The process proceeds as follows. Let $W_t$ be the set of weak vertices with respect to $C_t$. If $W_t = \emptyset$, then the process terminates. 
In this case we say that $C_t$ is a \emph{stable configuration}. 
Otherwise, we select with uniform probability a  permutation, $\sigma$,  of the elements of $W_t$.  
Considering the permutation as a sequence of the elements of $W_t$, we proceed through $|W_t|$ subrounds. 
At the $k^{th}$ subround we update the  $k^{th}$ vertex of the sequence, $v_k$, with respect to the current configuration (i.e., the configuration resulting from the $(k-1)^{th}$ subround). 
The configuration resulting from the $|W_t|^{th}$ subround is denoted $C_{t+1}$.
We refer to the process as \emph{the Prisoner's Dilemma process on $G$ with randomised asynchronous updating}. 
Though, for brevity we refer to this process as the \emph{PD process on $G$}.

Let $C$ be a configuration and $\sigma$ a permutation of the elements of $W$. For any vertex $x \in W$ we denote by $\sigma^{-1}(x)$ the position of $x$ in the sequence of elements of $W$ induced by $\sigma$. Thus, $x$ is updated in the $\sigma^{-1}(x)$ subround.

As we are interested in the spread of the cooperative strategy,  for any configuration we may consider the density of cooperators. 
For configuration $C_t$, let $r_t$ be the density of cooperators at time $t$.
If $C_t$ is a  stable configuration, then we define the \emph{final density}, denoted $r_f$, to be $r_t$.

We give an example on the $6\times 6$ grid to highlight how the choice of $T$ for the process and the choice of the updating permutation in a particular round affect the spread of strategies. Consider the configuration given in Figure \ref{fig:exampleC0}. In our figures we use white squares for cooperators and grey squares for defectors.  

 \begin{figure}
 	\begin{center}
 		\begin{tabular}{ccc}
 			\begin{subfigure}{0.2\linewidth}\centering
 				\includegraphics[width = \textwidth]{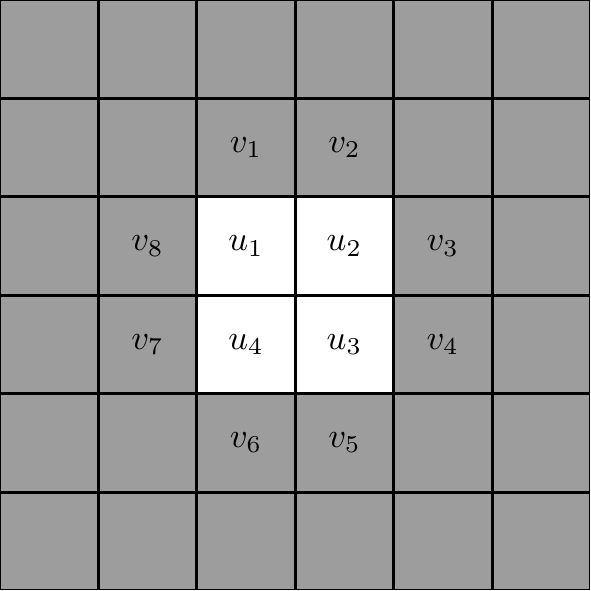}
 				\subcaption{$C_0$}
 				\label{fig:exampleC0}
 			\end{subfigure} &
 			\begin{subfigure}{0.41\linewidth}\centering
 				\includegraphics[width = \textwidth]{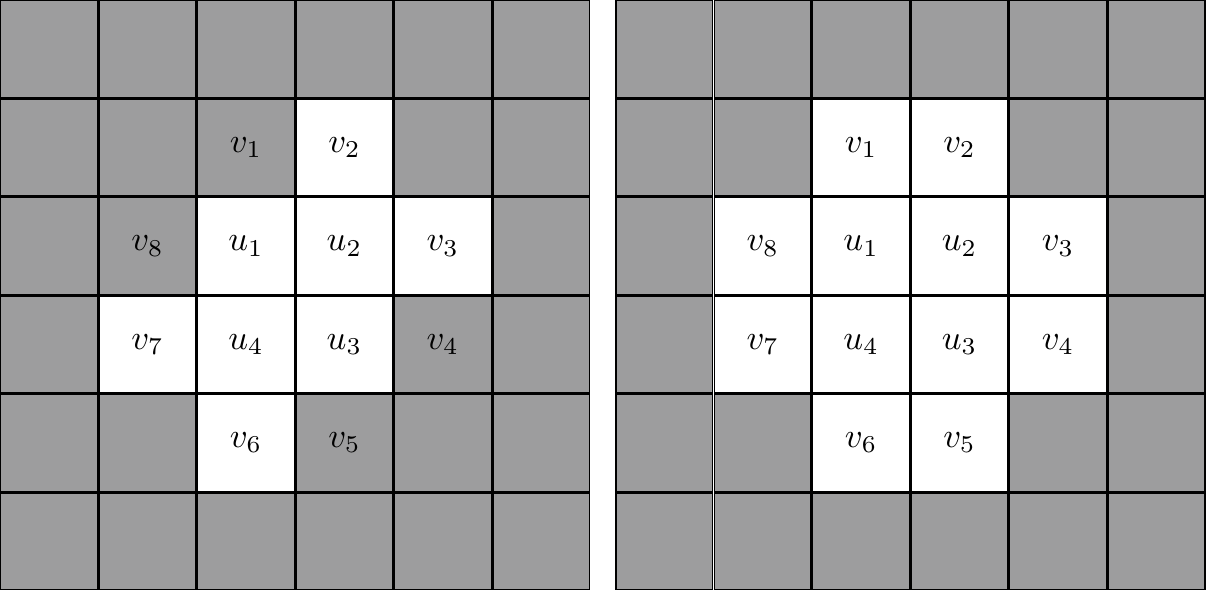}	
 				\subcaption{T = $\frac{5}{3}$}	
 				\label{fig:exampleC01}
 			\end{subfigure}
 			&
 			\begin{subfigure}{0.2\linewidth}\centering
 				\includegraphics[width = \textwidth]{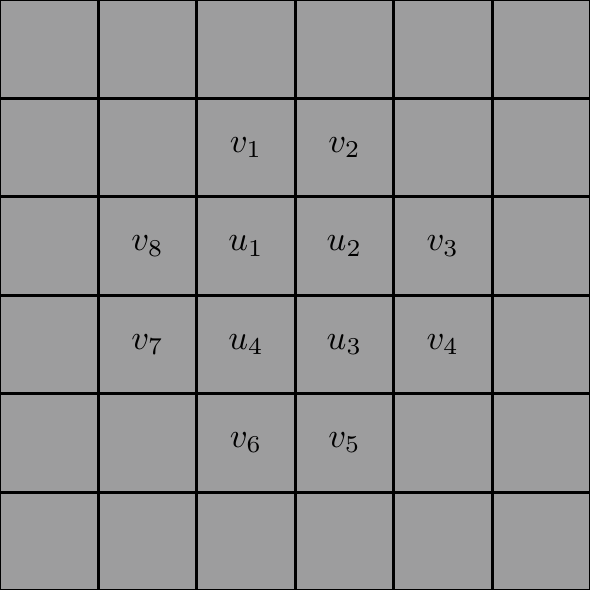}	
 				\subcaption{T = $\frac{8}{3}$}	
 				\label{fig:exampleC02}
 			\end{subfigure}
 			\\
 		\end{tabular}
 	\end{center}	
 	\caption{An example of the Prisoner's Dilemma Process with Random Asynchronous Updating}
 \end{figure}

If $T  = \frac{5}{3}$ each of the cooperators have score $2$ and each of the labelled defectors have score $\frac{5}{3}$. All unlabelled vertices have score $0$, as they are defectors with no cooperator neighbours. Observe that $W_0 = \{v_1,v_2,\dots, v_8\}$. Figure \ref{fig:exampleC01} gives the resulting configurations after applying $\sigma_1 = (v_2,v_3,v_6,v_7,v_1,v_4,v_5,v_8)$ to $C_0$ and alternatively applying $\sigma_2 = (v_2,v_4,v_6,v_8,v_1,v_3,v_5,v_7)$ to $C_0$. In the first case,  after subround $4$,  $v_1$ is no longer a weak vertex, and so does not change strategy. The configuration $C_1$ has no weak vertices and thus is stable. However, in the second case, $v_1$ is a weak vertex after subround $4$, and so does change from being a cooperator to a defector. In this second case, the resulting configuration has eight weak vertices.

For $T = \frac{8}{3}$, each of the cooperators have score $2$, and each of the labelled defectors have score $\frac{8}{3}$. All unlabelled vertices have score $0$, as they are defectors with no cooperator neighbours. Observe that $W_0 = \{u_1,u_2,u_3,u_4\}$. Regardless of the choice $\sigma$, Figure \ref{fig:exampleC02} is the resulting configuration after round $0$.

A configuration $C_t$ is called \emph{forced} if $C_{t+1}$ will be the configuration regardless of the choice of $\sigma$ at time $t$. For a sequence $C_0, C_1, \dots$ of configurations, we call the sequence resulting from removing the forced configurations, and re-indexing, the \emph{basic sequence}. We use the notation $C_0^\prime, C_1 ^\prime, \dots$  to refer to a basic sequence and use the term \emph{basic time steps} to refer to the time steps in a basic sequence.

In a $4$-regular graph, if $1 < T < \frac{4}{3}$, then the most successful neighbour of $v$ is the vertex in the closed neighbourhood of $v$ with the most cooperator neighbours, with defectors taking precedence in the case of a draw. For the remainder of this paper we consider only the case $1 < T < \frac{4}{3}$ as we restrict our study to the toroidal grid. We use the notation $T = 1 + \epsilon$ to refer to $T$ in this range and use $k + \epsilon$ to refer to a score between $k$ and $k+1$. Thus, we say that a defector with $k$ cooperator neighbours has score $k+\epsilon$.

\begin{center}
	\begin{figure}
	\begin{tabular}{c|c|c}
		\begin{subfigure}{0.33\textwidth}\centering
				\includegraphics[width=0.75\linewidth]{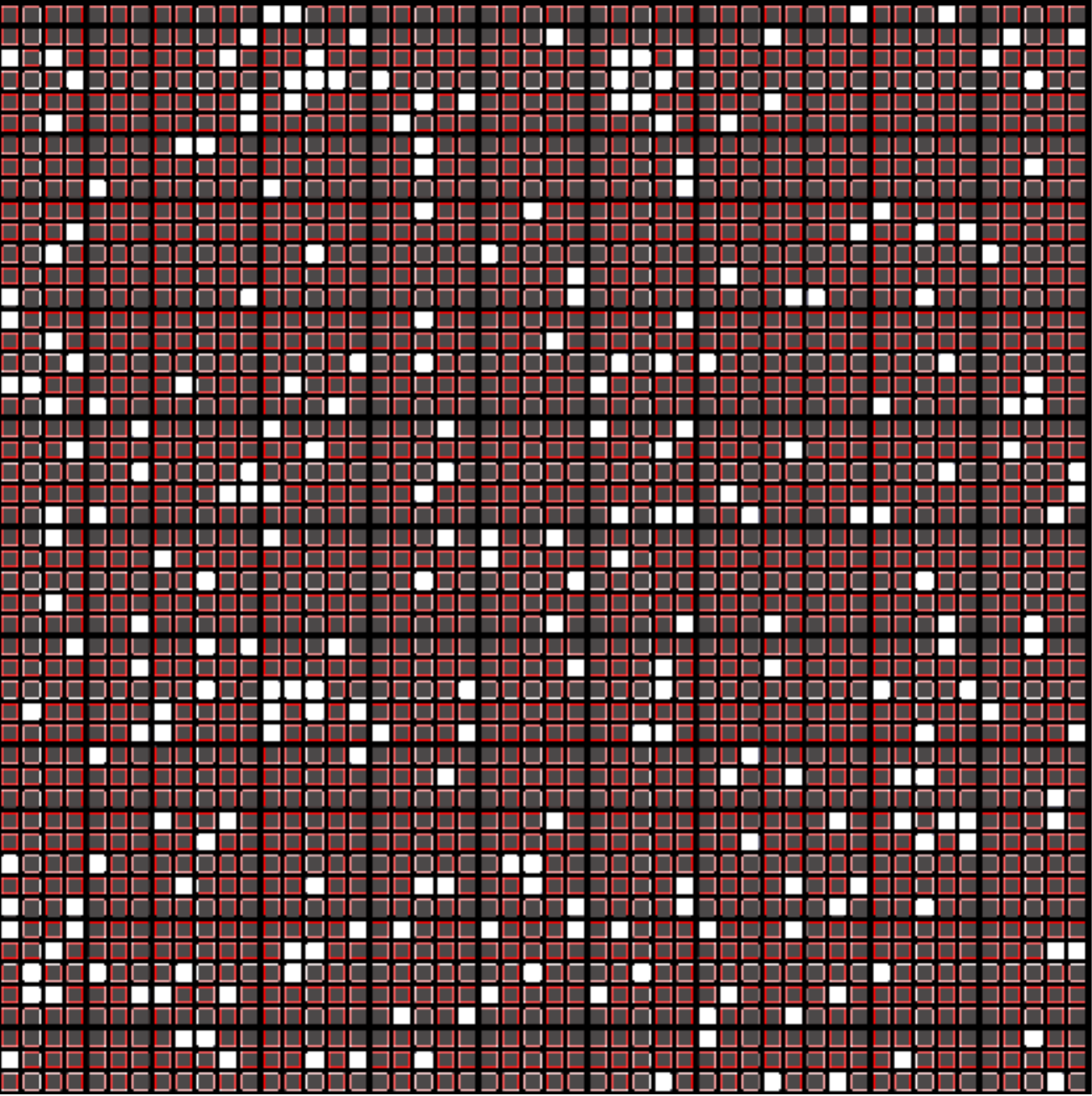}
			%	\caption{ .}
			%	\label{fig:smallp1}	
		\end{subfigure} &
			\begin{subfigure}[]{0.33\textwidth}\centering					
								\includegraphics[width=0.75\linewidth]{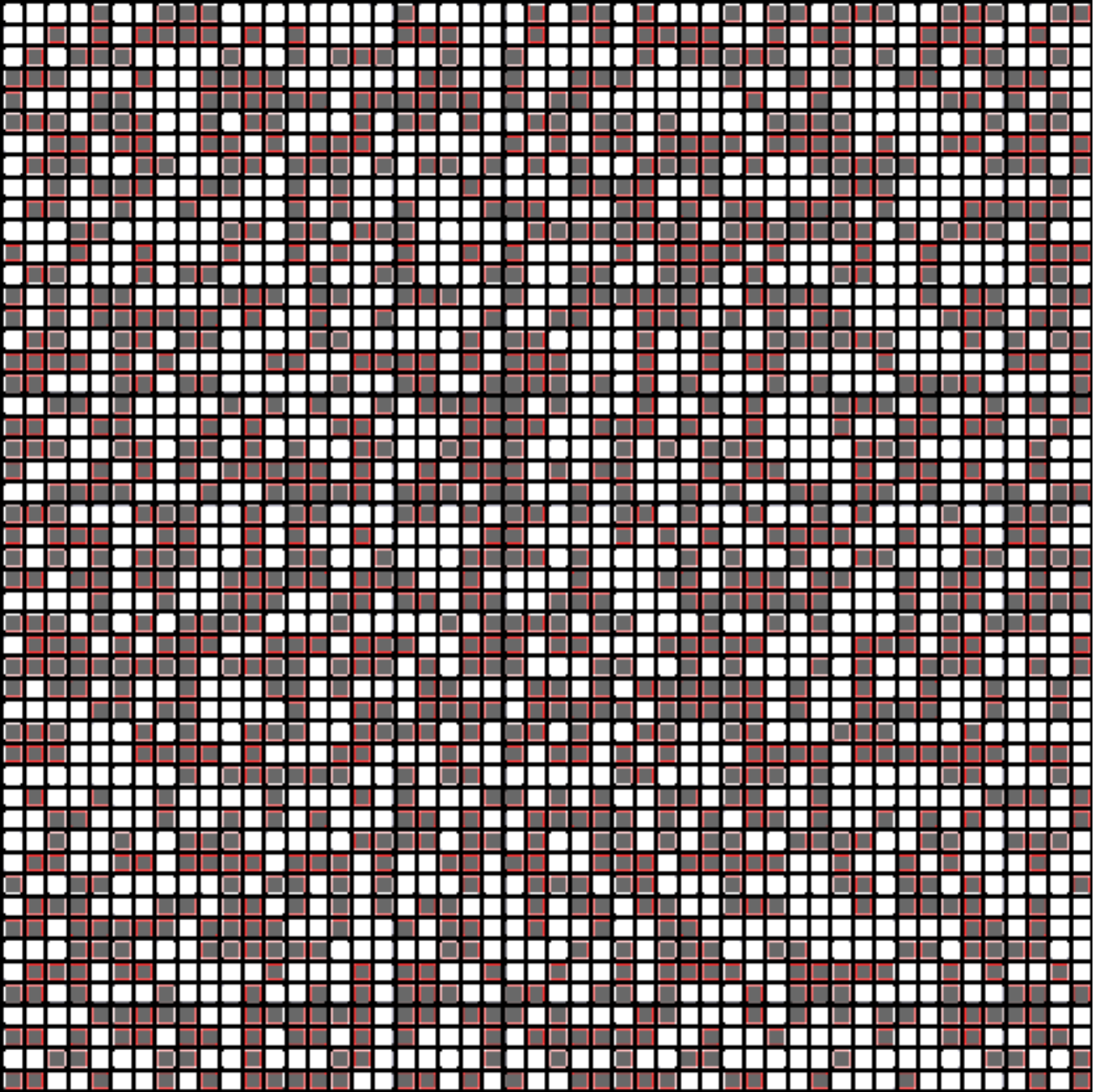}
							%	\caption{$C_0$,  $p =0.5$.}
							%	\label{fig:midp1}	
		\end{subfigure} &
		
		\begin{subfigure}{0.33\textwidth}\centering														
							\includegraphics[width=0.75\linewidth]{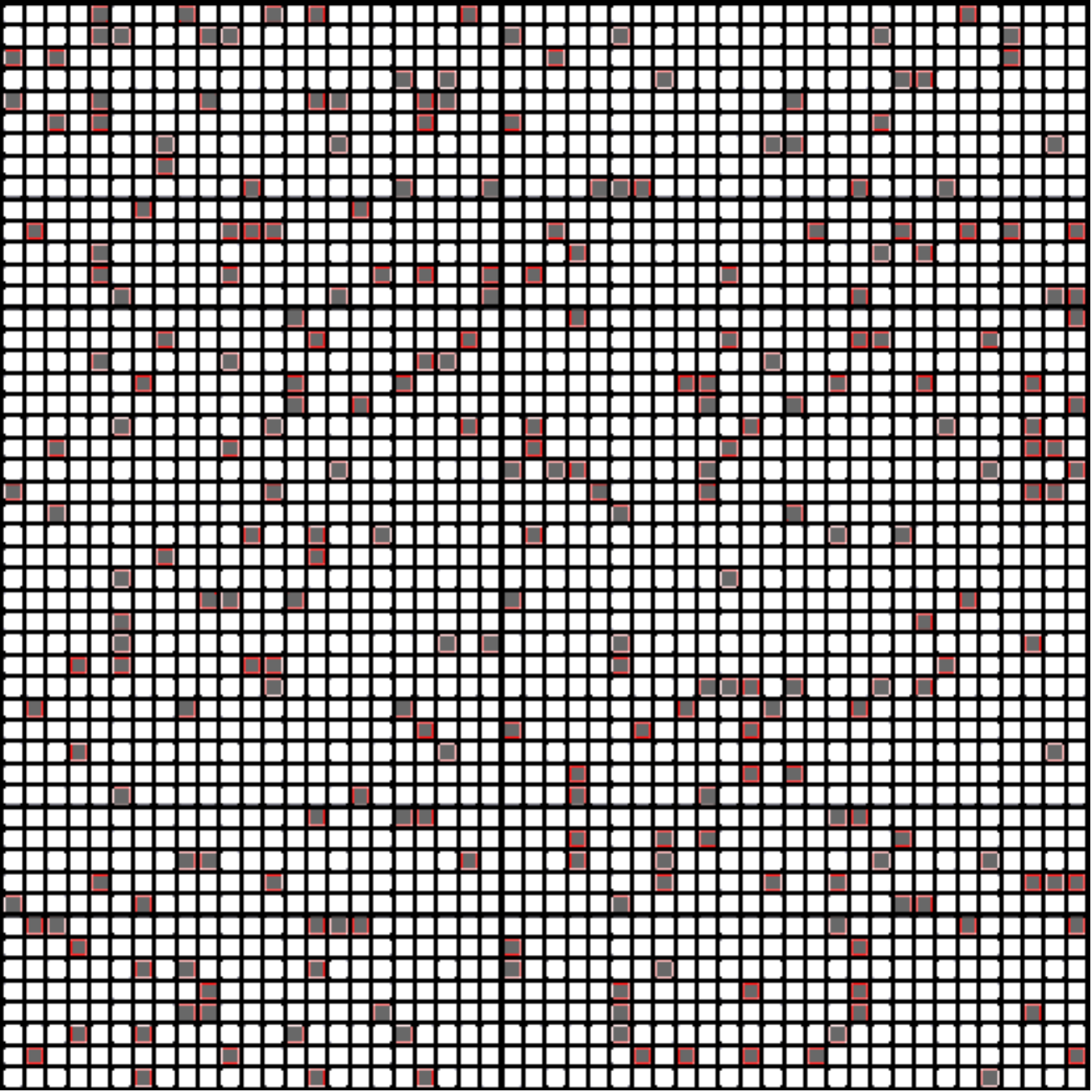}
						%	\caption{ $C_0$,  $p =0.9$.}
						%	\label{fig:largep1}	
			\end{subfigure} \\
		$C_0$,  $p =0.1$ & $C_0$,  $p =0.5$ & $C_0$,  $p =0.9$ \\
			\begin{subfigure}{0.33\textwidth}\centering					
				\includegraphics[width=0.75\linewidth]{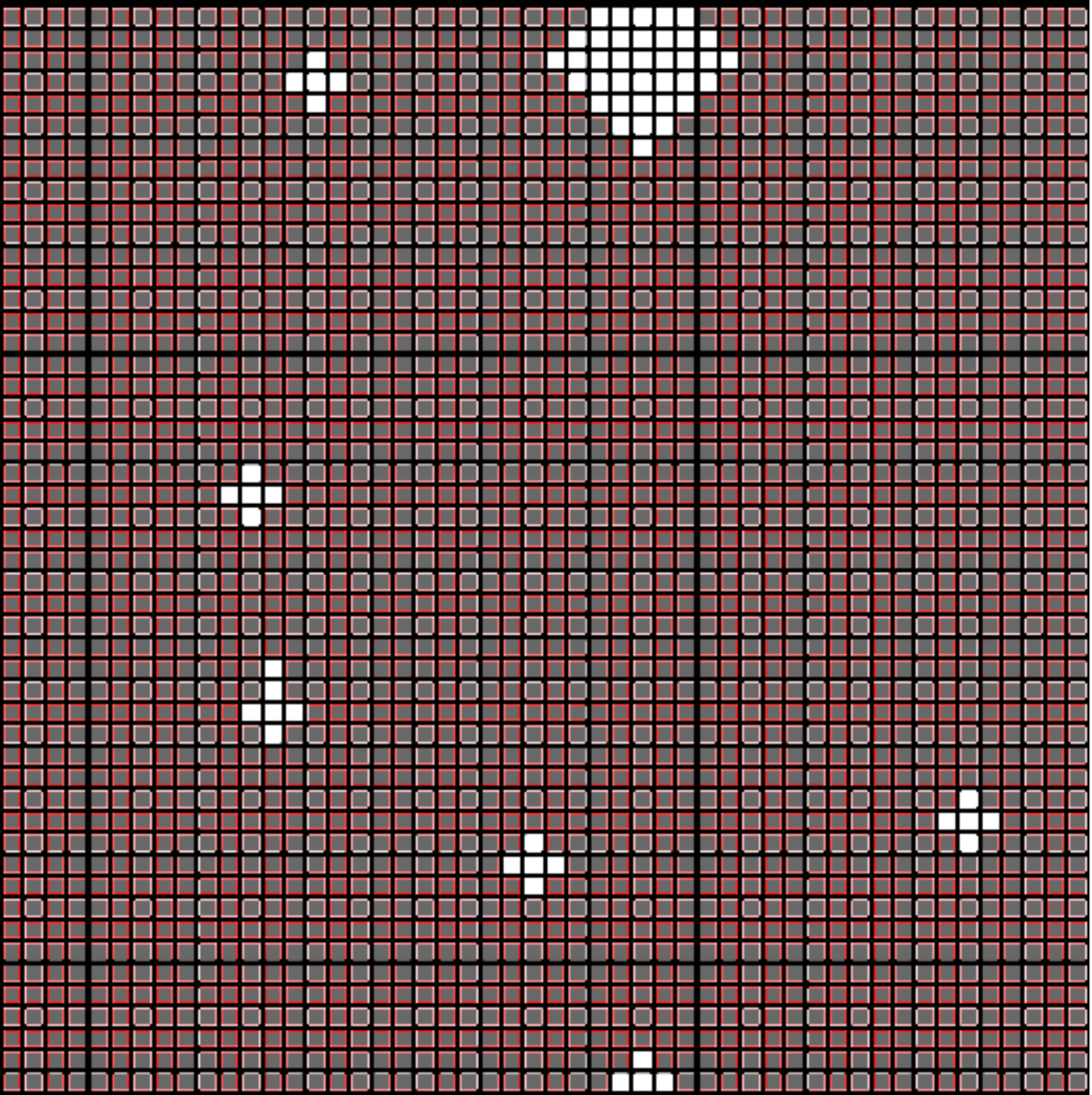}
			%	\caption{ $C_f$,  $p =0.1$.}
			%	\label{fig:smallp2}		
			\end{subfigure}&
			
					\begin{subfigure}{0.33\textwidth}\centering
				\includegraphics[width=0.75\linewidth]{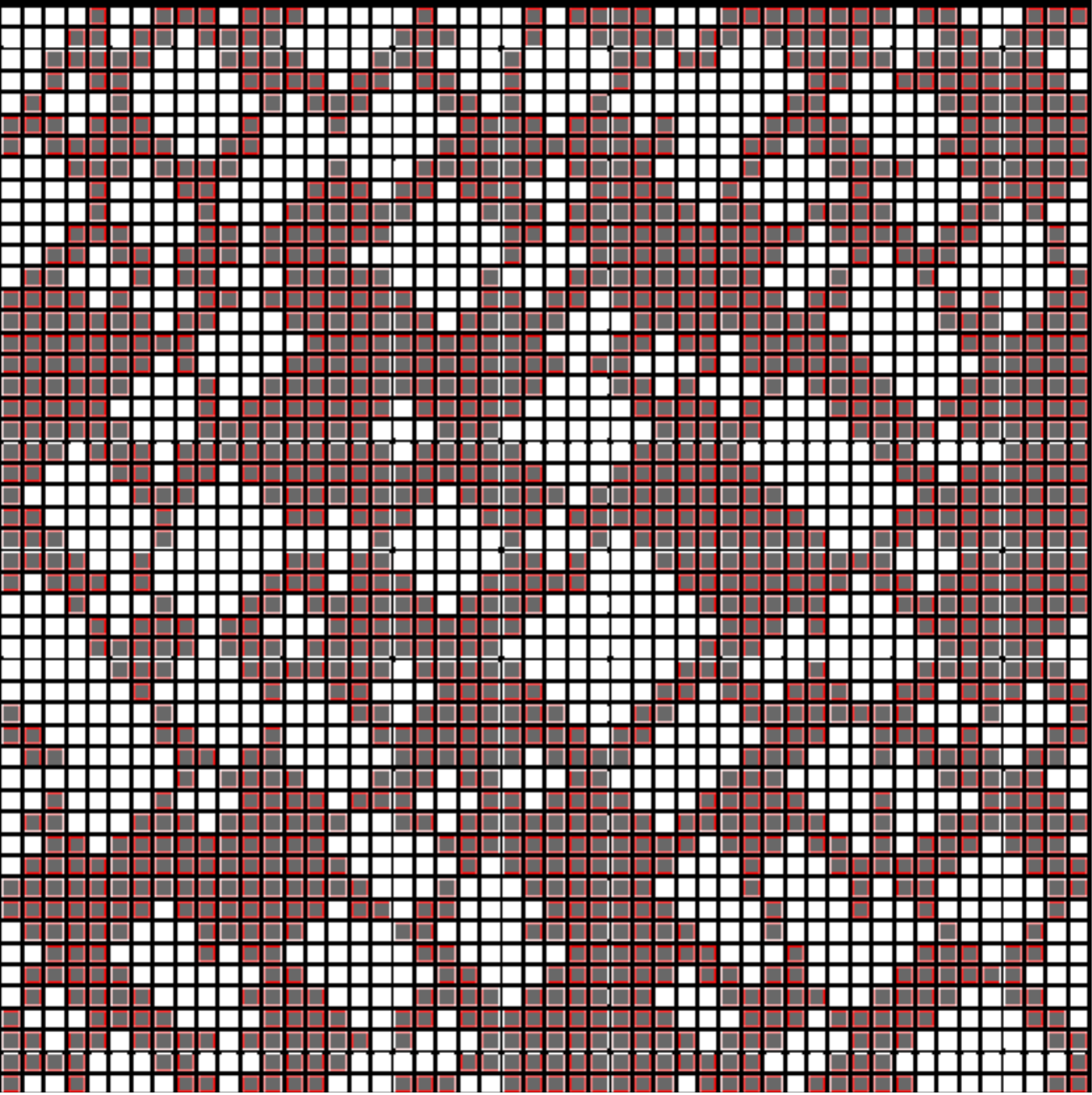}
				%\caption{ $C_f$,  $p =0.5$.}
			%	\label{fig:midp2}	
					\end{subfigure} &				
				\begin{subfigure}{0.33\textwidth}\centering 						
				\includegraphics[width=0.75\linewidth]{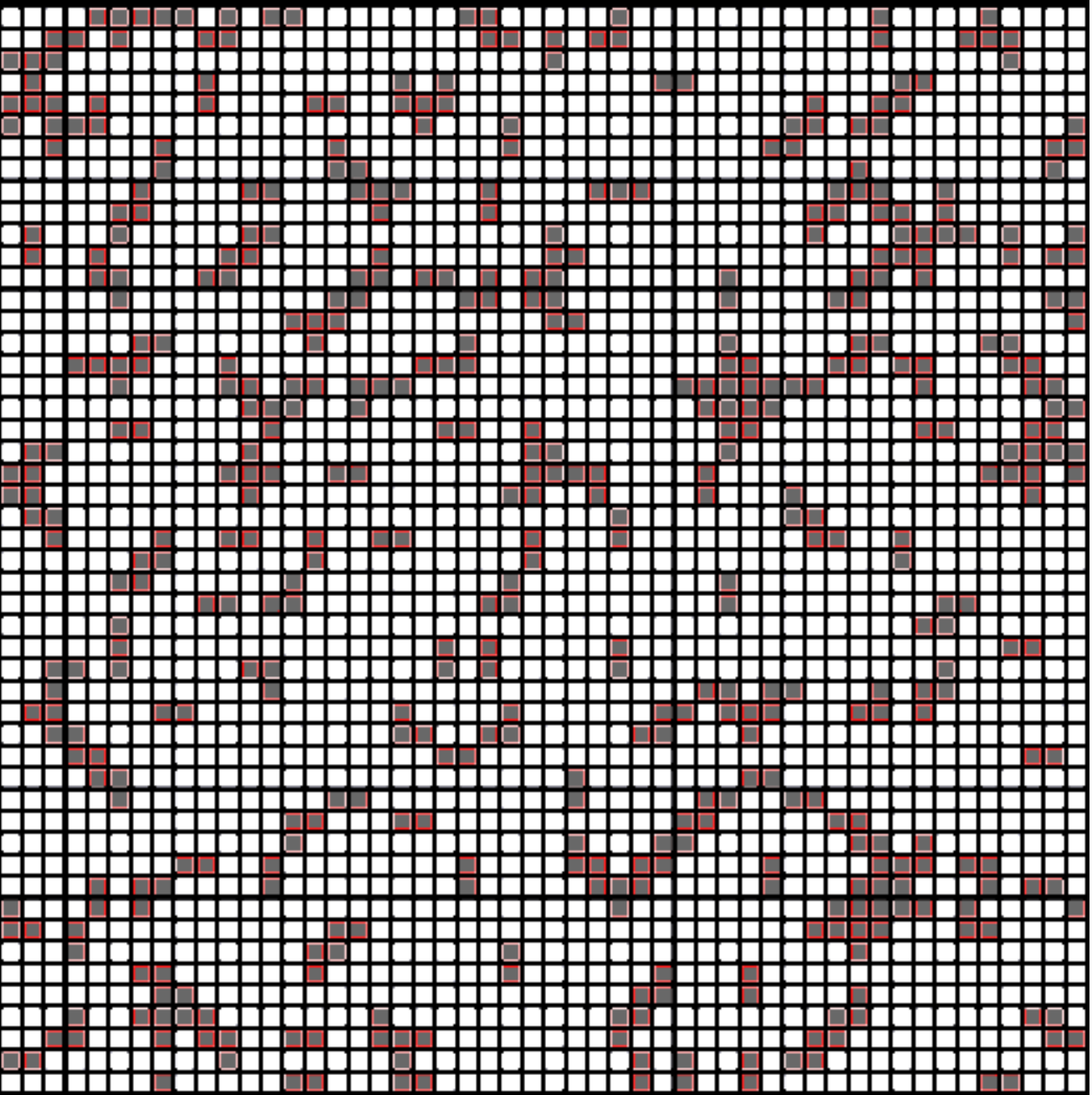}
				%\caption{ $C_f$,  $p =0.9$.}
			%	\label{fig:largep2}	
				\end{subfigure} \\
					$C_f$,  $p =0.1$ & $C_f$,  $p =0.5$ & $C_f$,  $p =0.9$ \\
			\end{tabular}
		\caption{}
		\label{fig:pexample}
	\end{figure}
\end{center}

\section{Evolution of  Clusters of Cooperators and Defectors in Infinite Grids} \label{sec:evolve}
In this section we consider the evolution of clusters situated in infinite grids. We  apply these results in Section \ref{sec:Growth} to study the behaviour of $r_t$ in the $n\times n$ toroidal grid. 

Up to rotation and reflection of the plane, there is a single $1$-cluster of defectors and a single $2$-cluster of defectors. A configuration of a single defector in an infinite grid of cooperators has exactly four weak vertices -- the neighbours of the single defector. By examining the number of cooperator neighbours, we see that when one of these weak vertices has become a cooperator the resulting configuration is stable. Therefore a $1$-cluster of defectors in an infinite field of cooperators evolves to a $2$-cluster of defectors, which is a stable configuration. To show that the spread of a $k$-cluster of defectors in an infinite grid of cooperators is bounded, we require the following results.

	\begin{proposition} \label{prop:noDefect}
		For any sequence $\{C_t\}_{t \geq 0}$ of configurations, no configuration $C_t$ with $t> 0$ contains an isolated defector.
	\end{proposition}
	
	\begin{proof}
		It suffices to show that if $C_0$ contains an isolated defector, then $C_1$ does not, and if $C_k$ does not contain an isolated defector, then $C_{k+1}$ does not. 
		
		If $v$ is an isolated defector in $C_0$, then  $s(v) = 4T = 4 + \epsilon$. Since $4 + \epsilon$ is the maximum score that can be attained by any vertex, any such $v$ is strong and each of the four cooperators adjacent to such a  $v$ are weak. Therefore during the first round, at least one vertex adjacent to $v$ will become a defector.

	% 	It follows that each vertex in the neighbourhood of $v$ is an element of $W_0$. Since $v \notin W_0$ we must show that at least one of the neighbours of $v$ is a defector at time $t =1$. Let $\sigma$ be a permutation of the elements of $W_0$. For a vertex $x$, let $\sigma^{-1}(x)$ be index of $x$ in the permutation induced by $\sigma$. Let $u$ be the element of $N(v)$ with the earliest index in $\sigma$.  Let  $\sigma^{-1}(u) = i$, it follows that before sub-round $i$ that $v$ is a isolated defector. Therefore $u$ is weak when given the opportunity to change strategy during subround $i$. Since $v \notin W_0$ and $v$ has a defector neighbour at the end of round $1$, it cannot be that $v$ is an isolated defector.
		
		Assume now that $C_k$ does not contain an isolated defector and that $v$ is an isolated defector in $C_{k+1}$. We proceed  with cases based on the strategy of $v$ at the start of the round.
		
		\emph{Case 1: $v$ is a cooperator in $C_k$:}
		Let $\sigma$ be a permutation of $W_k$. Let $x \in N[v]$ such that $x$ is the last vertex to change strategy of all vertices in $N[v]$. If $x=v$, then all neighbours are cooperators in the subround  that $v$ changes. However, in this case $v$ is not weak at the start of subround $\sigma^{-1}(x)$.
		
		If $x \in N(u)$, then it turned from defector to cooperator. 
		In subround $\sigma^{-1}(x)$ it must be that the score of $v$ is $3 + \epsilon$. 
		No neighbour of $x$ that is a cooperator can have a neighbour that scores greater than $3$. This contradicts that $x$ changes. 
		
		\emph{Case 2: $v$ is a defector in $C_k$:} 
		Let $x \in N(u)$ such that $x$ is the last vertex to change strategy of all vertices in $N(u)$. This follows similarly to the previous case.	
	\end{proof}
		
	For a configuration $C_i$, we say that a cooperator $v$ is a \emph{persistent cooperator} if $C_t(v) = 1$ for all $t \geq i$.
	
	\begin{proposition} \label{prop:4strong}
		For any sequence $\{C_t\}_{t \geq 0}$ of configurations, if a cooperator vertex $v$ has four cooperator neighbours in $C_t$, then $v$ is a persistent cooperator.
	\end{proposition}
	
	\begin{proof}
		By Proposition \ref{prop:noDefect}, $s_t(u) \leq 4$ for all $t > 0$ and all $u \in V$, as cooperators have score at most $4$ and non-isolated defectors have score no more than $3+\epsilon$. If $s_t(v) = 4$, then $v$ is a cooperator with four cooperator neighbours. This implies that each vertex of $N[v]$ is strong. Therefore if $s_t(v) = 4$, then $s_{t+1}(v) = 4$. 
	\end{proof}
	
	\begin{corollary} \label{cor:initialPersist}
			If  $v$ is a cooperator with no isolated defector at distance at most $2$ in $C_0$, then $v$ is a persistent cooperator.
	\end{corollary}
	
	\noindent In this case we say that $v$ is a \emph{initial persistent cooperator}
	
	Together these results allow us to find a bound on the growth of a $k$-cluster of defectors in an infinite field of cooperators.
	
	\begin{corollary}\label{cor:defectnogrow}
		If $C_0$ is the configuration of the infinite grid consisting of a $k$-cluster of defectors in a field of cooperators  such that the  $k$-cluster is contained by a rectangle of length $\ell$ and width $w$, then there exists a rectangle of length  $\ell+4$  and width $w+4$ so that the growth of the defector strategy is contained within this rectangle.
	\end{corollary}
	
	\begin{proof}
		Every cooperator at distance two from the cluster of defectors is an initial persistent cooperator. 
	\end{proof}
	
As we consider the evolution of $k$-clusters of cooperators in an infinite field of defectors we encounter some cases for which there are no surviving cooperators. In this case, we say that the particular cluster evolves to an \emph{empty cluster}.

Up to symmetry of the plane, there is a single $1$-cluster and a single $2$-cluster. When placed in a sufficiently large grid of defectors, each of these clusters evolves to an empty cluster after at most two time steps.

Up to rotation and reflection there are two species of $3$-clusters: \emph{$3$-lines} and \emph{$3$-corners}. When placed in a sufficiently large field of defectors,  a $3$-line evolves to a stable configuration containing a $5$-cluster with probability $1$.  A $3$-corner evolves to a stable configuration containing a $5$-cluster with probability $\frac{1}{2}$ and to an empty cluster with probability $\frac{1}{2}$. The evolution of these clusters is given in Figure \ref{fig:3cluster}. Note that weak vertices are indicated with a circle.

Up to rotation and reflection there are $5$ species of $4$-clusters: \emph{$4$-lines, $4$-corners, $4$-hats, $4$-turns, } and \emph{$4$-squares} (See Figure \ref{fig:4clusters}). A $4$-hat stabilises to a stable $5$-cluster with probability $1$. However, for each of the other configurations simulation suggests a non-zero probability of large growth. The evolution of these clusters through a small number of iterations is given in Figure \ref{fig:4clusters}. 

\begin{center}
	\begin{figure}\centering
		\begin{tabular}{ccc}
			\begin{subfigure}{0.33\textwidth}\centering
				\includegraphics[width = \linewidth]{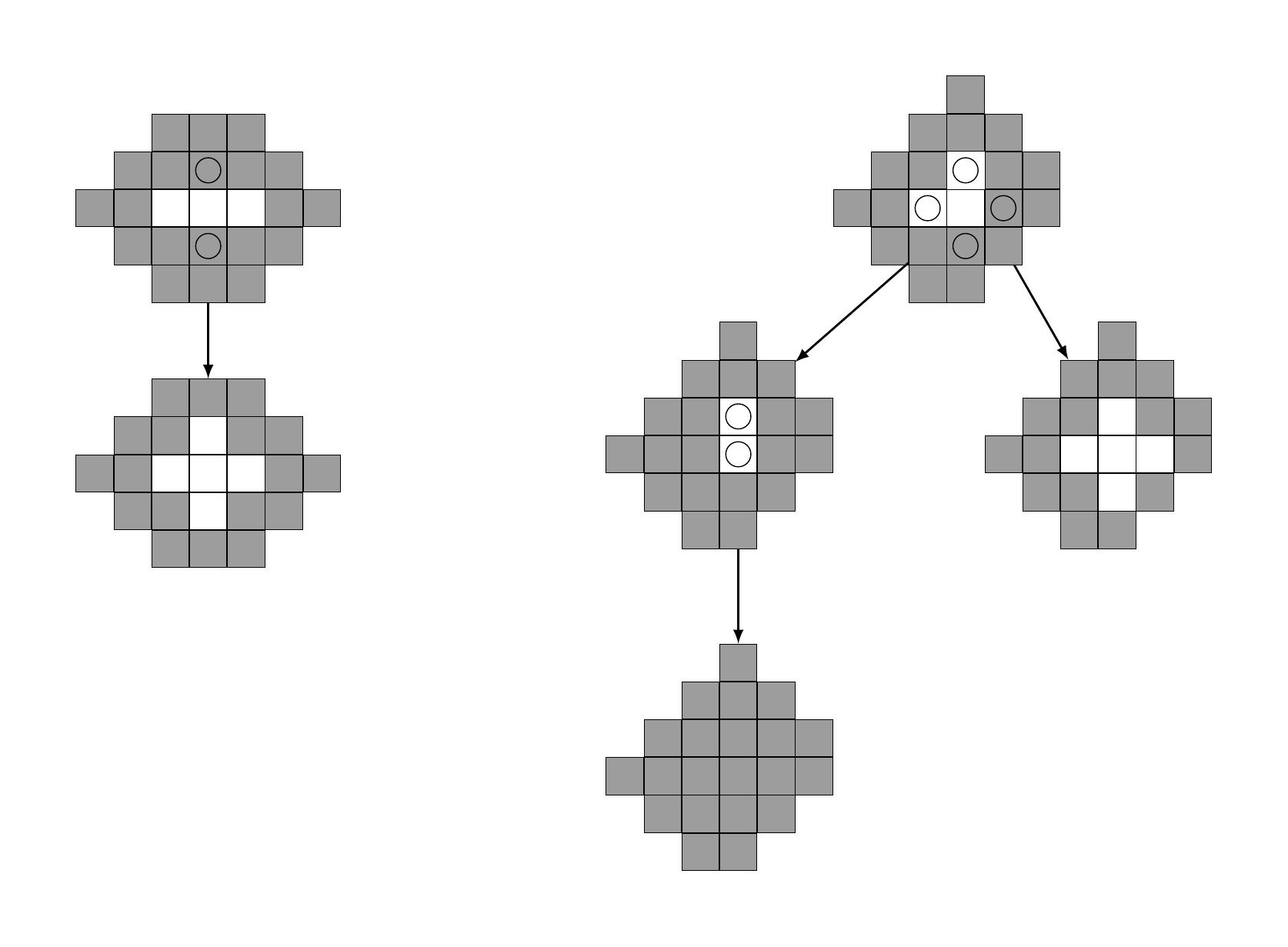}
				\caption{ $3$-lines and $3$-corners}
				\label{fig:3cluster}	
			\end{subfigure} &
			\begin{subfigure}{0.33\textwidth}\centering
				\includegraphics[width=\linewidth]{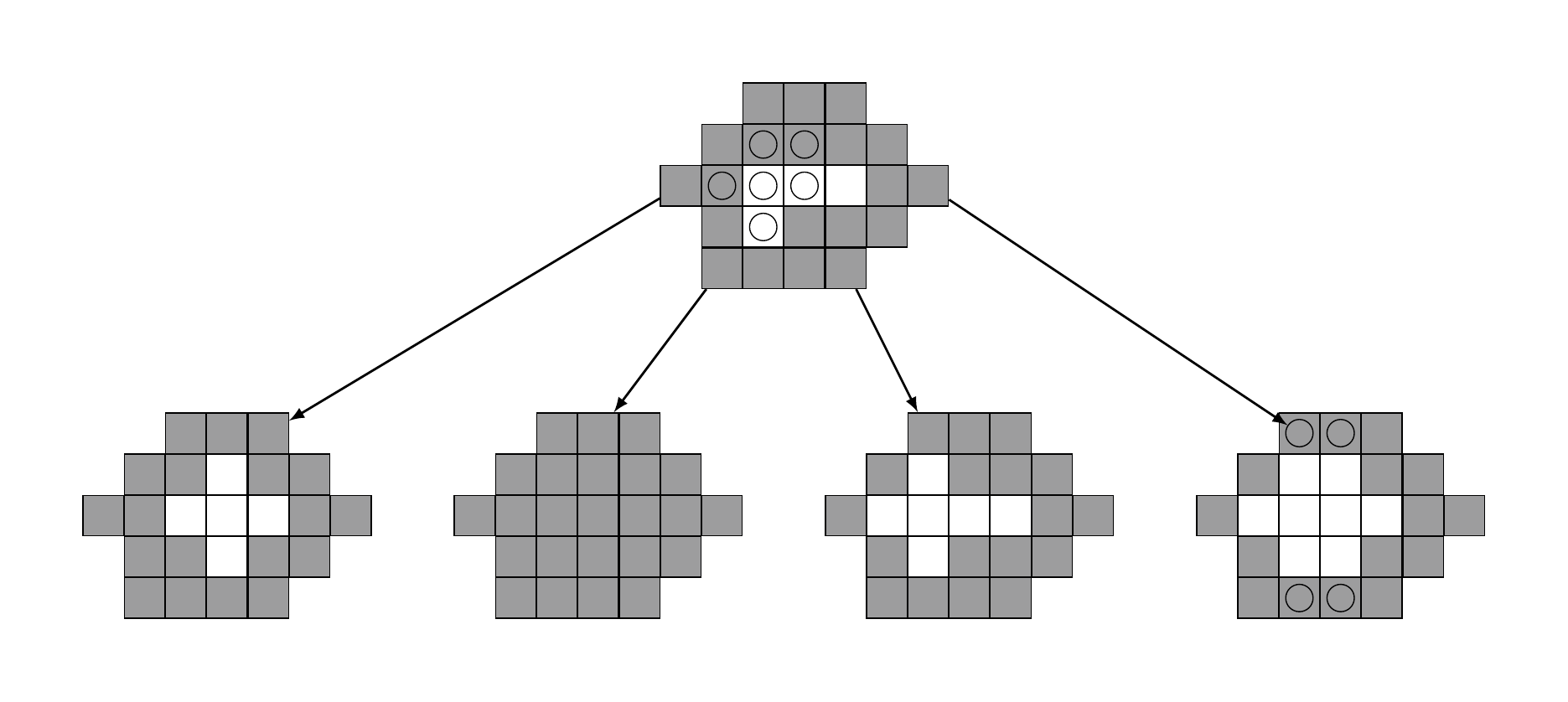}
				\caption{ $4$-corner}
				\label{fig:4corner}	
			\end{subfigure} &
			\begin{subfigure}{0.33\textwidth}\centering					
				\includegraphics[width=\linewidth]{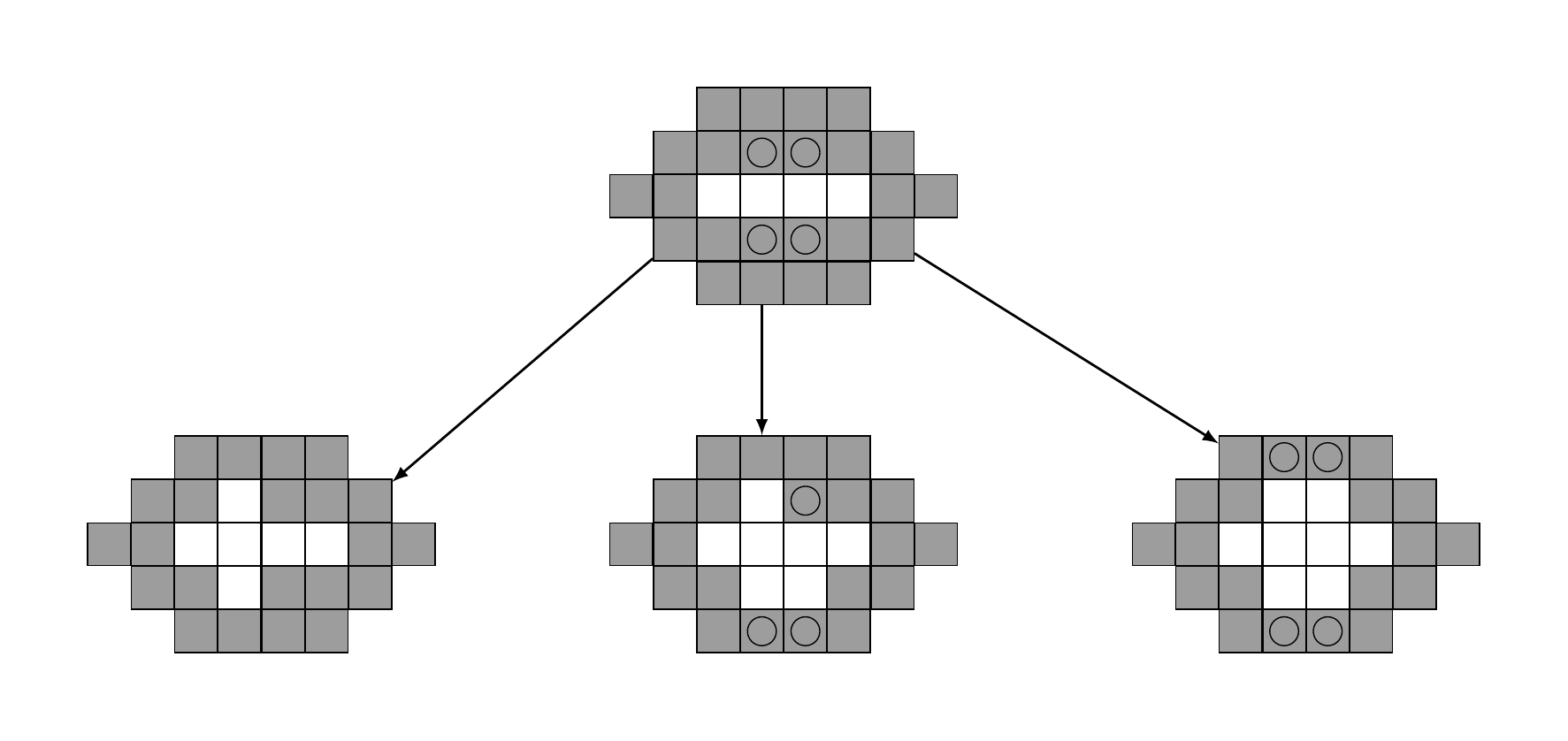}
				\caption{ $4$-line}
				\label{fig:4line}	
			\end{subfigure} \\			
			\begin{subfigure}{0.33\textwidth}\centering														
				\includegraphics[width=0.5\linewidth]{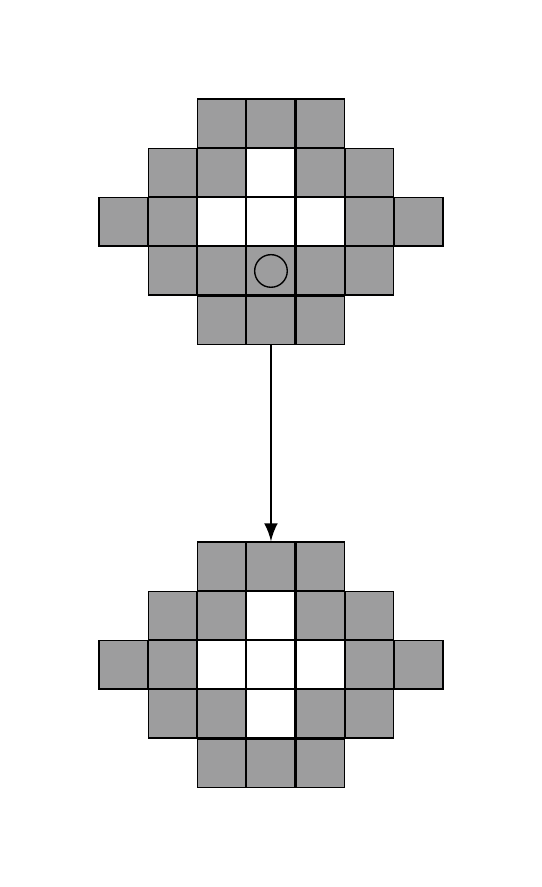}
				\caption{$4$-hat}
				\label{fig:4hat}	
			\end{subfigure} &
			\begin{subfigure}{0.33\textwidth}\centering					
				\includegraphics[width=\linewidth]{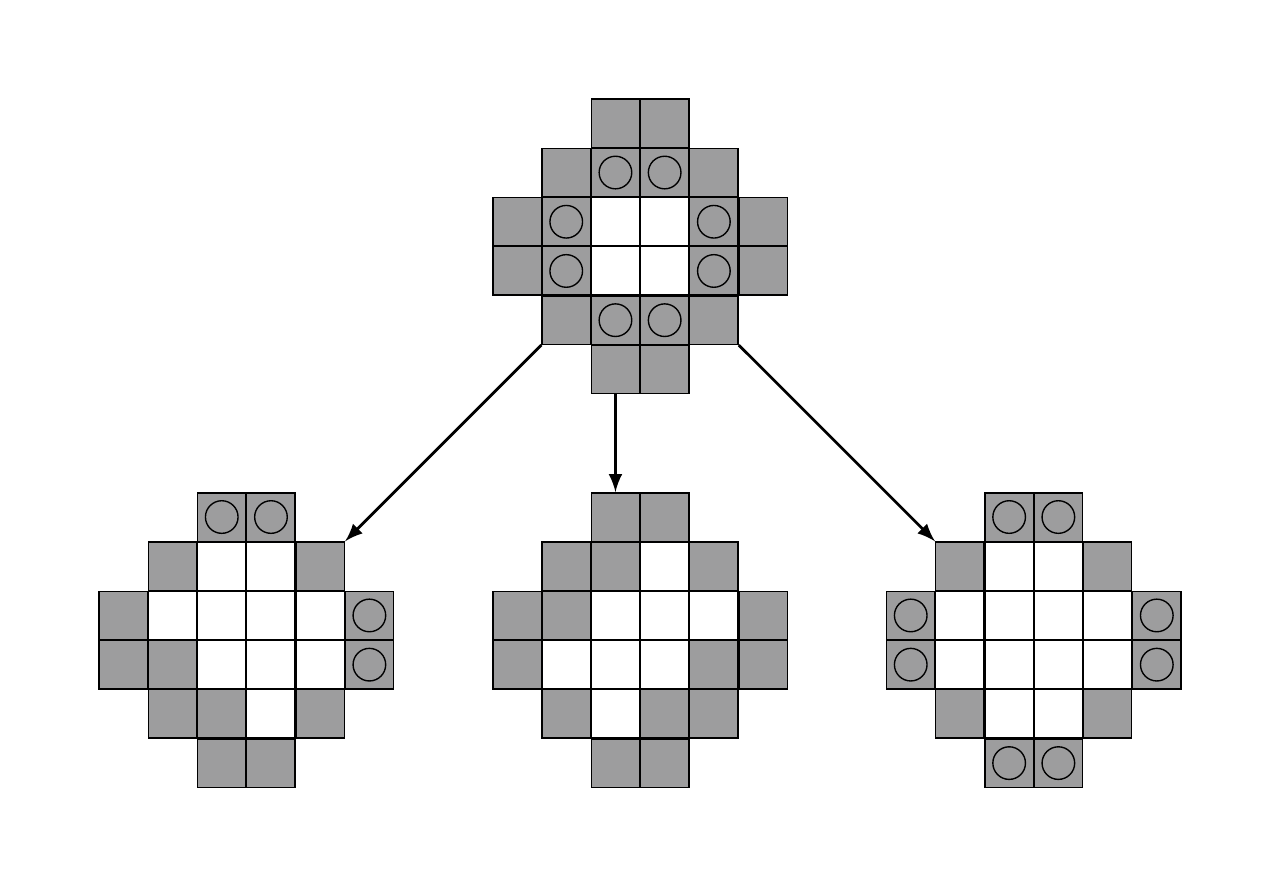}	
				\caption{ $4$-square}
				\label{fig:4square}		
			\end{subfigure}&			
			\begin{subfigure}{0.33\textwidth}\centering
				\includegraphics[width=\linewidth]{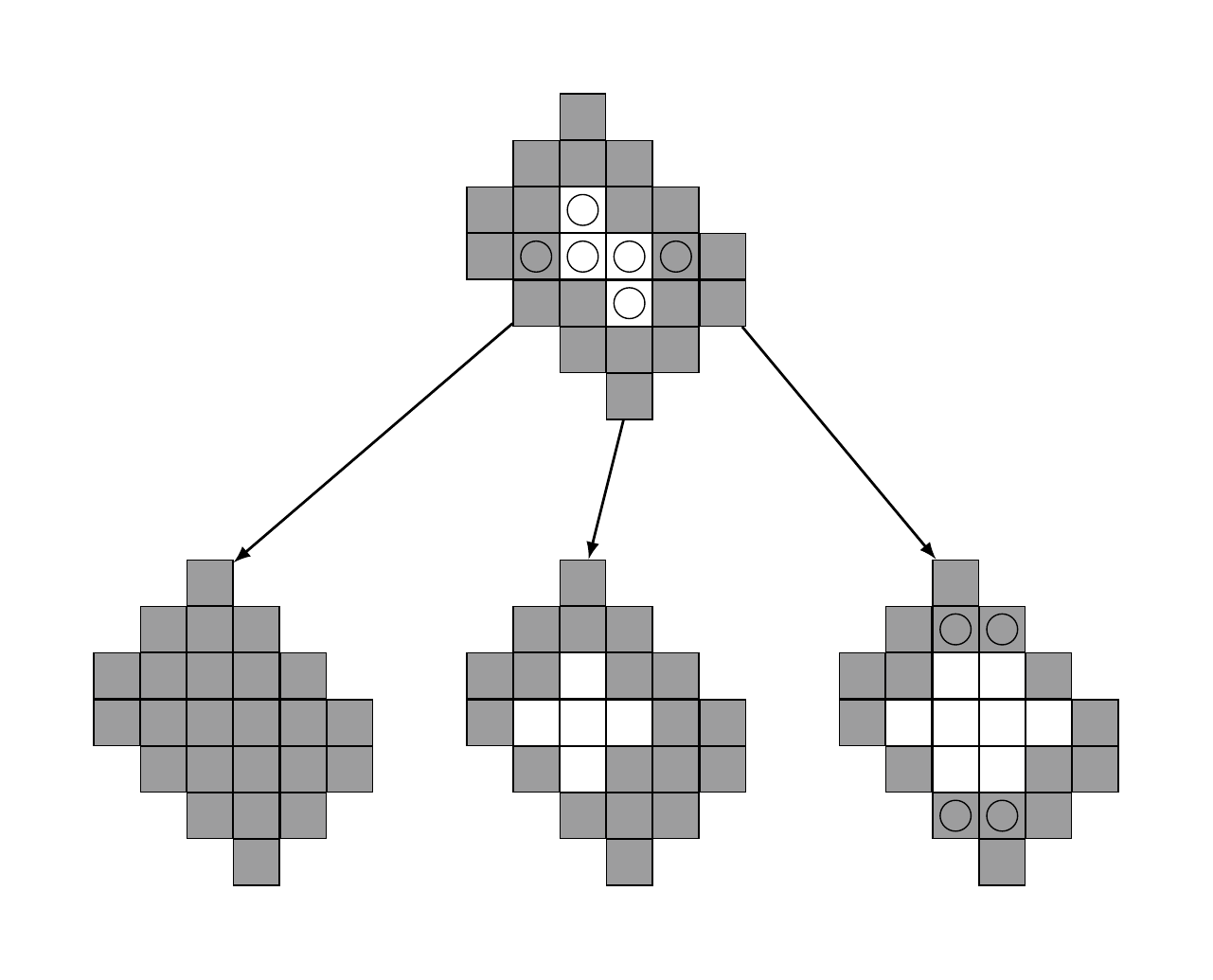}
				\caption{ $4$-turn}
				\label{fig:4-turn}	
			\end{subfigure} 	\\
		\end{tabular}
		\caption{Evolution of $3$-clusters and $4$-clusters of cooperators in an infinite field of defectors.}
		\label{fig:4clusters}
	\end{figure}
\end{center}

We wish to show that a $4$-cluster in an infinite grid of defectors will eventually evolve to a stable configuration. 
Though the growth of such clusters passes through many different configurations, we show that every configuration in the sequence of basic configurations where $C_0$ is a $4$-cluster of cooperators in an infinite field of defectors can be classified into one of eight types. We consider these types equivalent under rotation and reflection. The width parameter of each type tells the number of columns that contain cooperators, whereas the height parameter tells us how many cooperators are in the column with the greatest amount of cooperators. A summary of the configurations is given in Table \ref{tab:configs}. We note that each of these configurations is convex. That is, for any pair of cooperators contained in a cluster, there is a shortest path between them that contains only cooperators.

We begin by defining four basic clusters, which take the form of skew rectangles.

\begin{table}\centering
	\begin{tabular}{|c|c|c|c|}
		\hline 
		\textbf{Name} & \textbf{Glyph}  
		%& \textbf{Column Heights}  
		& \textbf{No. Weak Vertices}  
		%& \textbf{No. Collaborators}  
		\\ 
		\hline 
		Stable Cluster & $\stable_{w,h}$   
		%& $1,\dots 2h-3, 2h-1,$ \newline $\dots, 2h-1, 2h-3, \dots, 1$  
		& $0$  
		%& $2h^2$+ \newline $(2h-1)(2w - 2h - 1)$
		\\ 
		\hline 
		Doubly Even Cluster & $\flatfour_{w,h}$   
		%& $2, \dots, 2h-2, 2h-1, \dots$ \newline  $2h-1, 2h-3, \dots, 1$  
		& $8$  
		%&  $2{ h \choose 2} + (h-1)^2 $ \newline $+ (2h-1)(2w-2h-2)$
		\\ 
		\hline
		Adjacent Even Cluster & $\flatadj_{w,h}$  
		%& $2, \dots, 2h-2, 2h-1,$ \newline  $\dots, 2h-1, 2h-3,$ \newline $ \dots, 3,1$  
		& $4$  
		%&  $2{ h \choose 2} + (h-1)^2 $ \newline $+ (2h-1)(2w-2h-2)$
		\\ 
		\hline
		Opposite Even Cluster & $\flatopp_{w,h}$    &$4$
		%&  $1, \dots, 2h-3, 2h-1, \dots,$ \newline  $2h-1, 2h-3, \dots,1$
		%  & $4$  & $2h^2 +$ \newline $ (2h-1)(2w - 2h)$ 
		\\ 
		\hline
		Double Even Transit Cluster & $\flatoppb_{w,h,\ell}$ &   $3$\\
		\hline
		Adjacent Even Transit Cluster A  & $\flatadjb_{w,h,\ell}$ &   $3$\\
		\hline
		Adjacent Even Transit Cluster B  & $\flatoppa_{w,h,\ell}$ &   $3$\\
		\hline
		Adjacent Even Transit Cluster C  & $\flatadja_{w,h,\ell}$  &  $3$\\		
		\hline
	\end{tabular} 
	\caption{Summary of Defined Configurations}
	\label{tab:configs}
\end{table}

The \emph{stable cluster of width $w\geq 3$ and height $h\geq 3$} $(h \equiv 1 \bmod 2, w-h\geq 0)$  consists of cooperators in $w$ columns. 
The first column (starting from the left) contains a single cooperator. 
The number of cooperators increases by $2$ in each subsequent column, until a column of height $h$ is reached.
These increasing columns are aligned  so that the previous column is vertically centred in the subsequent column.
After the maximum is reached, columns of height $h$  repeat $w-h \geq 0$ times so that there are $w$ columns containing cooperators, as shown in Figure \ref{fig:4stable}. 
These repeating columns of height $h$ are offset in such a way so that  a subsequent column is  one unit lower than the previous column
Following these repeating columns are columns whose heights decrease by $2$ in each subsequent column until there is a column with a single cooperator. 
An example is given in Figure \ref{fig:4stable}. Observe that such a configuration is stable when placed in an infinite field of defectors. 
We use the symbol $\stable_{w,h}$ to denote a  stable cluster of width $w$ and height $h$.  

The \emph{doubly even cluster of height $h \geq 2$ and width $w \geq 2$}  $(h \equiv 1 \bmod 2, w-h+1 \geq 0)$ is a cluster consisting of cooperators in $w$ columns. 
The first column (starting from the left) contains two cooperators. 
The number of cooperators increase by $2$ in each subsequent column, until a column of height  $h-1$ is reached.
These increasing columns are aligned so that the previous column is vertically centred in the subsequent column. 
If $w-h+1 = 0$, then there a second column of height $h-1$ aligned with the first column of height $h-1$.
The heights of the columns then decrease by $2$ in each subsequent column until there is a column with a single cooperator. 
Otherwise if $w-h+1 > 0$, then the column of height $h-1$ is followed by  $w-h+1$ columns of height $h$.
The first of these columns is aligned so that the top of this column aligns with the column of height $h-1$.
Subsequent columns of height $h$ are offset so a subsequent column is one unit lower than the previous column.
Following these repeating columns is a column of height $h$ there is a column of height $h-1$.
This column of height $h-1$ is aligned so the bottom of this column is aligned with that of the previous column.
 The heights of the columns then decrease by $2$ in each subsequent column until there is a column with a two cooperators. 
An example is given in Figure \ref{fig:4flatfour}. 
We use the symbol $\flatfour_{w,h}$  to denote the  doubly even cluster of height $h$ and width $w$. Observe that $\flatfour_{w,h}$ has $8$ weak vertices when placed in an infinite field of defectors.  

The \emph{opposite even cluster of height $h\geq 3$ and width $w\geq 4$}  $(h \equiv 0 \bmod 2, w-h\geq 0)$ is a cluster of cooperators in $w$ columns.
 The first column (starting from the left)  contains one cooperator. 
 The number of cooperators increases by $2$ in each subsequent column, until $h-1$ is reached. 
 These increasing columns are aligned so  that the previous column is vertically centred in the subsequent column. 
 If $w-h = 0$, then the column of height $h-1$ is followed by a second column of height $h-1$.
 This second column of height $h-1$ is aligned with the previous column of height $h-1$.
The heights of the columns then decrease by $2$ in each subsequent column until there is a column with a single cooperator. 
Otherwise if $w-h>0$, then the column of height $h-1$ is followed $w-h$ columns of height $h$.
The first of these is aligned so that the top of this column is aligned with the top of the previous column of height $h-1$.
Subsequent columns of height $h$ are offset so a subsequent column is are one unit lower than the previous column.
These columns of height $h$ are followed by a column of height $h-1$.
This column of height $h-1$ is aligned with the bottom of the previous column of height $h$.
 The heights of the columns then decrease by $2$ in each subsequent column until there is a column with a single cooperator.
An example is given in Figure \ref{fig:4flatopp}.
 We use the symbol $\flatopp_{w,h}$ to denote the opposite even cluster of height $h$ and width $w$. 
 Observe that $\flatopp_{w,h}$ has four weak vertices when placed in an infinite field of defectors.

The \emph{adjacent even cluster of height $h\geq 5$ and width $w\geq 5$} $(h \equiv 0 \bmod 2)$  is a cluster of cooperators in $w$ columns.  
The first column (starting from the left)  contains two cooperators.
 The number of cooperators increases by $2$ in each subsequent column, until a column of height $h$ is reached. 
 These increasing columns are aligned so that so that the previous column is vertically centred in the subsequent column. 
 These increasing columns are followed by $w-h$ columns of height $h$ so that there are $w$ columns containing cooperators, as shown in Figure \ref{fig:4flatadj} . 
 These repeating columns of height $h$ are offset so that a subsequent column is are one unit lower than the previous column.
 The repeating columns of height $h$ are followed by a column of height $h-1$.
 This bottom of the column of height $h-1$ is aligned with the bottom of the final column of height $h$.
 The heights of the columns then decrease by $2$ in each subsequent column until there is a column with a single cooperator. 
 An example is given in Figure \ref{fig:4flatadj}. 
 We use the symbol $\flatadj_{w,h}$  to denote the adjacent even cluster of height $h$ and width $w$. Observe that $\flatadj_{w,h}$ has four weak vertices when placed in an infinite field of defectors. 
 
 \begin{center}
 	\begin{figure}
 		\begin{tabular}{cccc}
 			\begin{subfigure}{0.2\textwidth}\centering
 				\includegraphics[width=\linewidth]{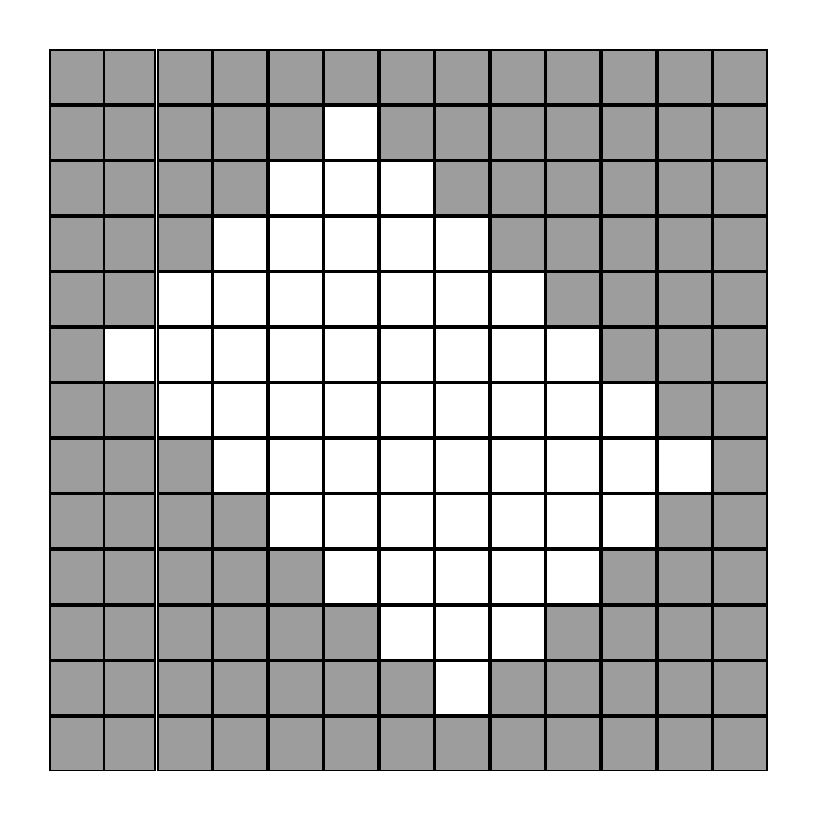}
 				\caption{$\stable_{11,9}$} 
 				\label{fig:4stable}	
 			\end{subfigure} &
 			\begin{subfigure}{0.2\textwidth}\centering					
 				\includegraphics[width=\linewidth]{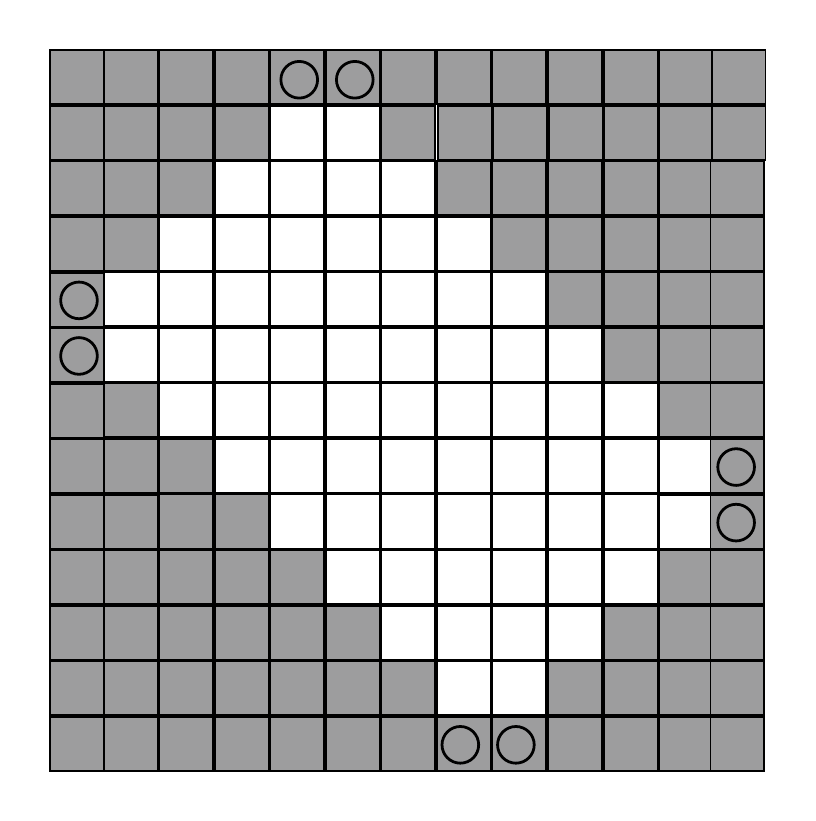}
 				\caption{$\flatfour_{11,9}$}
 				\label{fig:4flatfour}	
 			\end{subfigure} &	
 			\begin{subfigure}{0.2\textwidth}\centering					
 				\includegraphics[width=\linewidth]{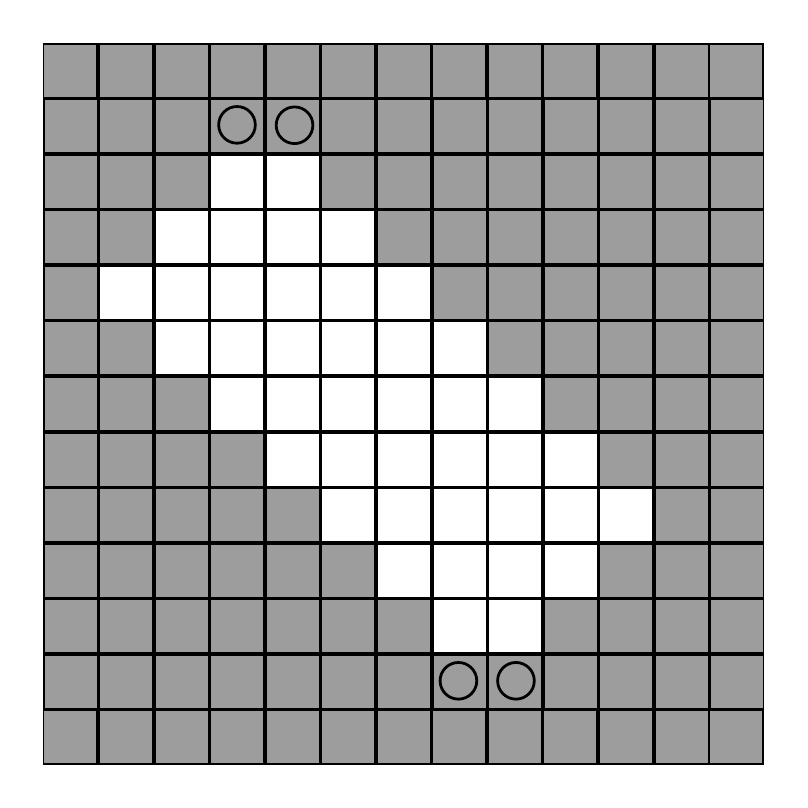}
 				\caption{$\flatopp_{10,6}$}
 				\label{fig:4flatopp}	
 			\end{subfigure} &
 			\begin{subfigure}{0.2\textwidth}\centering					
 				\includegraphics[width=\linewidth]{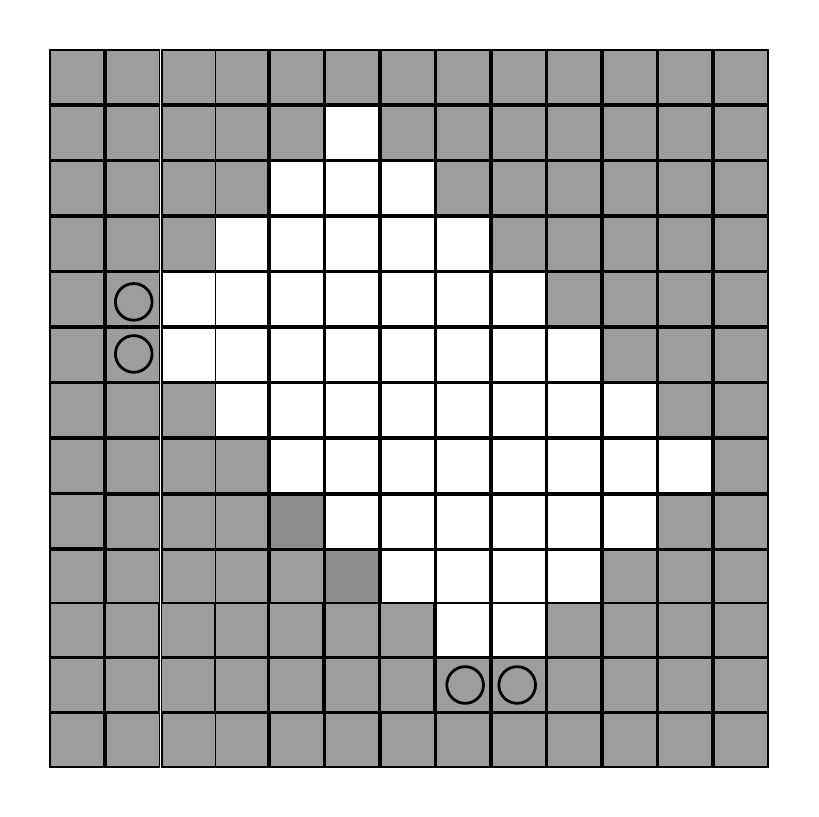}
 				\caption{$\flatadj_{10,8}$}
 				\label{fig:4flatadj}		
 			\end{subfigure}\\
 		\end{tabular}
 		\caption{Basic  Clusters}
 		\label{fig:growClusterA}
 	\end{figure}
 \end{center}
 
From these last two skew rectangular clusters, we define four transit clusters. These four transit clusters arise by considering the evolution of the skew rectangular clusters placed in an infinite field of defectors. Consider the cluster given in Figure \ref{fig:4flatadj}. We see a pair of weak vertices aligned vertically on the left side of the figure and a pair of weak vertices aligned horizontally at the bottom of the figure. In the subsequent round exactly one defector from each pair will turn to a collaborator. Assume that the upper weak vertex on the left side of the figure and the left vertex on the bottom of the figure turn to collaborators in the following round. In this case we see that the next three time-steps are forced, as the $NW$ and $SW$ side of the cluster grow. Since the $NW$ side is shorter than the $SW$ side, when the next basic time-step is reached the $SW$ side has not yet been fully filled in. We can consider such a cluster as being formed from a copy $\flatadj$ by adding some number of  collaborators along one of the sides. In a similar way, we can also construct a cluster from a copy of $\flatopp$ by adding collaborators along one of the sides.

The \emph{double even transit cluster of width $w$, height $h$  and length $\ell$}  ($1 \leq \ell < \frac{h}{2}$) is formed in two ways based on the value of $\ell$.
To describe this cluster, we must consider $\flatopp_{w,h}$  placed in an infinite field of defectors.
Such a cluster has two pairs of weak vertices -- those at the top of the cluster and those at the bottom of the cluster.
When $\ell=1$ the double even transit cluster of width $w$, height $h$  and length $\ell=1$ is formed by changing the upper-right weak defector to a collaborator.
We note that the double even transit cluster of width $w$, height $h$  and length $\ell=1$ has three weak vertices when placed in an infinite field of defectors.
One of these weak vertices, say $x$, has no weak neighbours.
When $1 < \ell < \frac{h}{2}$ we define this cluster inductively. 
The double even transit cluster of width $w$, height $h$  and length  $1 < \ell < \frac{h}{2}$  is formed from the the double even transit cluster of width $w$, height $h$  and length $ \ell-1$ by changing $x$ to be a cooperator.
We note that the double even transit cluster of width $w$, height $h$  and length $1 \leq \ell < \frac{h}{2}$ has three weak vertices when placed in an infinite field of defectors. 
One of these weak vertices, say $x$, has no weak neighbours.
We use the symbol $\flatoppb_{w,h,\ell}$ to denote the double even transit cluster of width $w$, height $h$ and length $\ell$. 

The following two transit clusters are formed from $\flatadj$.
Consider $\flatadj_{w,h}$ placed in an infinite field of defectors.
Such a cluster has two pairs of weak vertices -- those at the left of the cluster and those at the bottom of the cluster. The two following clusters arise by adding, respectively, collaborators along the $NW$ side and along the $SW$ side starting from the left side of the cluster.

The \emph{adjacent even transit cluster of width $w$, height $h$  and length $\ell$ of type A (resp. type B)} ($1 \leq \ell < h/2$)  is formed in two ways based on the value of $\ell$.
When $\ell=1$, the adjacent even transit cluster of width $w$, height $h$  and length $\ell=1$ of type A (resp. type B) is formed from $\flatadj_{w,h}$ by changing the lower (resp. upper) of the two weak vertices on the left side of the cluster to be collaborators.
Observe that such a cluster has three weak vertices.
One of these weak vertices, say $x$, has no weak neighbours.
For  $1 < \ell < \frac{h}{2}$ we define this cluster inductively.
The adjacent even transit cluster of width $w$, height $h$  and length  $1 < \ell < \frac{h}{2}$  of type A (resp. type B)  is formed from an adjacent even transit cluster of width $w$, height $h$  and length $\ell-1$ of type A (resp. type B) by changing $x$ to be a cooperator.
We note that the adjacent even transit cluster of width $w$, height $h$  and length $1 \leq \ell < \frac{h}{2}$ of type A (resp. type B) has three weak vertices when placed in an infinite field of defectors. 
One of these weak vertices, say $x$, has no weak neighbours.
We use the symbol $\flatadjb_{w,h,\ell}$ (resp.$\flatoppa_{w,h,\ell}$)  to denote the {adjacent even transit cluster of width $w$, height $h$ and length $\ell$ of type A}. See Figure \ref{fig:4flatadjb} for an example.

Finally, we reach our ultimate transit cluster -- the \emph{adjacent even transit cluster of width $w$, height $h$  and length $\ell$ of type C}  ($0 \leq \ell < h/2$, $w,h \equiv 0 \bmod 2)$. We describe this cluster using the same methodology as the basic clusters. 
This cluster consist of cooperators in $w$ columns.
We begin with the case $\ell = 0$.
The first column (starting from the left) contains a single cooperator. 
The number of cooperators increases by $2$ in each subsequent column, until a column of height $h-4$ is reached. 
These increasing columns are aligned so that so that the previous column is vertically centred in the subsequent column. 
This column of height $h-3$ is followed by a column of height $h-1$.
The column of height $h-3$ is aligned with that of height $h-1$ so that the bottom of the column of height $h-3$ is aligned with the vertex that is third from bottom in the column of height $h-1$.
The column of height $h-1$ is followed by a column of height $h$.
The bottom of the column of height $h-1$ is aligned with the bottom of the column of height $h$.
The column of height $h$ is then followed by $w-h-3$ columns of height $h$.
These repeating columns of height $h$ are aligned so that a subsequent column is one unit higher than a previous column.
Following the repeating columns of height $h$ are columns whose heights decrease by $2$ until a column of height $1$ is reached.
We note that adjacent even transit cluster of width $w$, height $h$  and length $\ell = 0$ of type C has three weak vertices when placed in an infinite field of defectors.
One of these vertices, say $x$, has no weak neighbours.
For  $1 \leq \ell < \frac{h}{2}$ we define this cluster inductively.
The adjacent even transit cluster of width $w$, height $h$  and length $1 \leq \ell < \frac{h}{2}$ of type C  is formed from  the adjacent even transit cluster of width $w$, height $h$  and length $\ell-1$  of type C by changing $x$ to be a cooperator.
We note that the adjacent even transit cluster of width $w$, height $h$  and length $1 \leq \ell < \frac{h}{2}$ of type C has three weak vertices when placed in an infinite field of defectors. 
One of these weak vertices, say $x$, has no weak neighbours.
We use the symbol $\flatadja_{w,h,\ell}$ to denote the adjacent even transit cluster of width $w$, height $h$ and length $\ell$ of type C. See Figure \ref{fig:4flatadja} for an example.

\begin{center}
	\begin{figure}\centering
		\begin{tabular}{cccc}			 			 
			\begin{subfigure}{0.2\textwidth}\centering														
				\includegraphics[width=\linewidth]{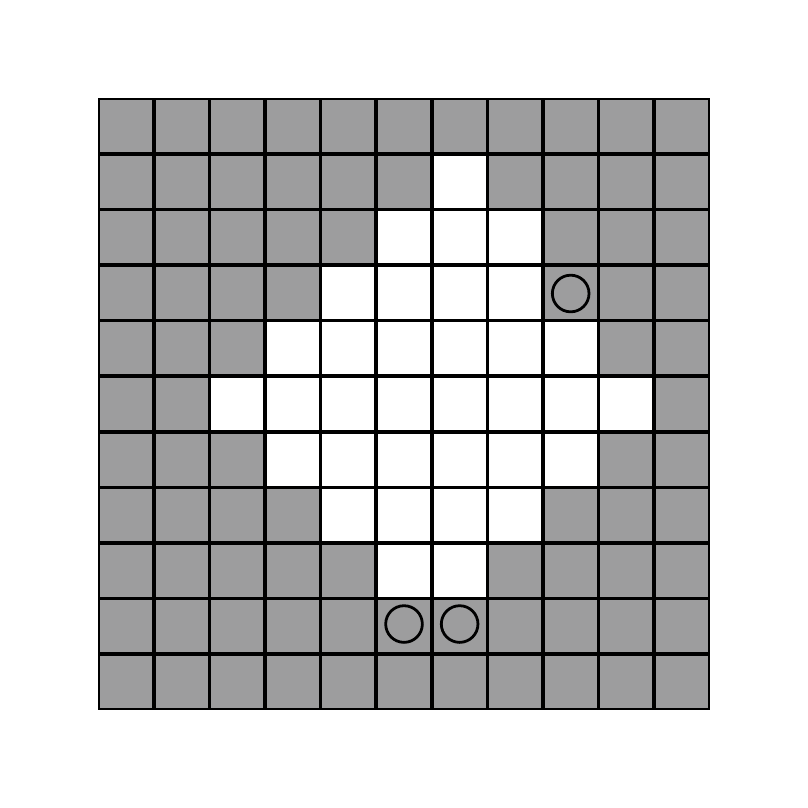}
				\caption{$\flatoppb_{8,8,2}$}
				\label{fig:4flatoppb}	
			\end{subfigure} &			
			\begin{subfigure}{0.2\textwidth}\centering
				\includegraphics[width=\linewidth]{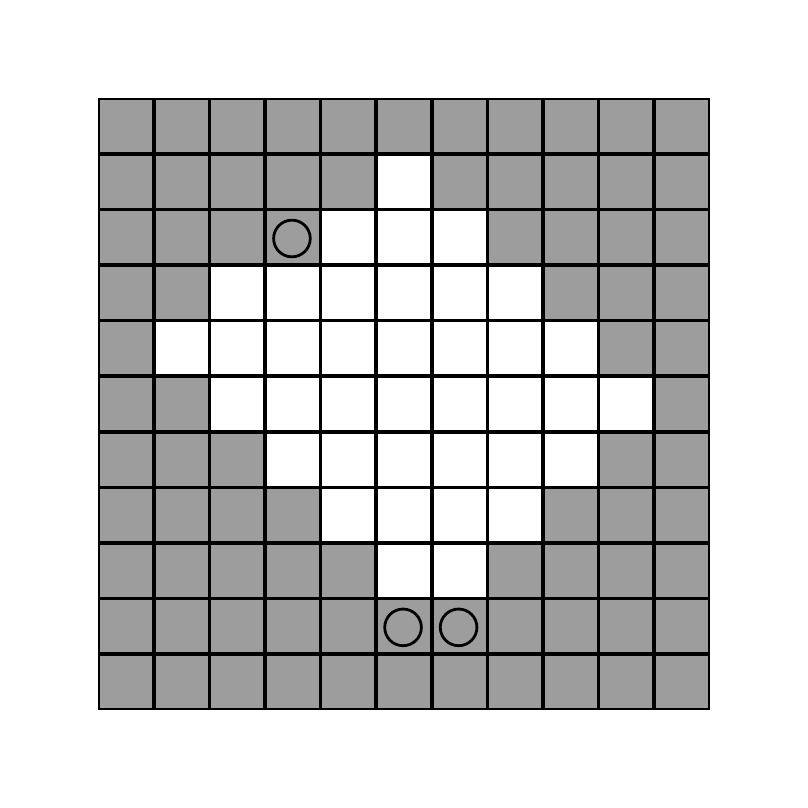}
				\caption{$\flatoppa_{9,8,2}$}
				\label{fig:4flatoppa}	
			\end{subfigure} &
			\begin{subfigure}{0.2\textwidth}\centering					
				\includegraphics[width=\linewidth]{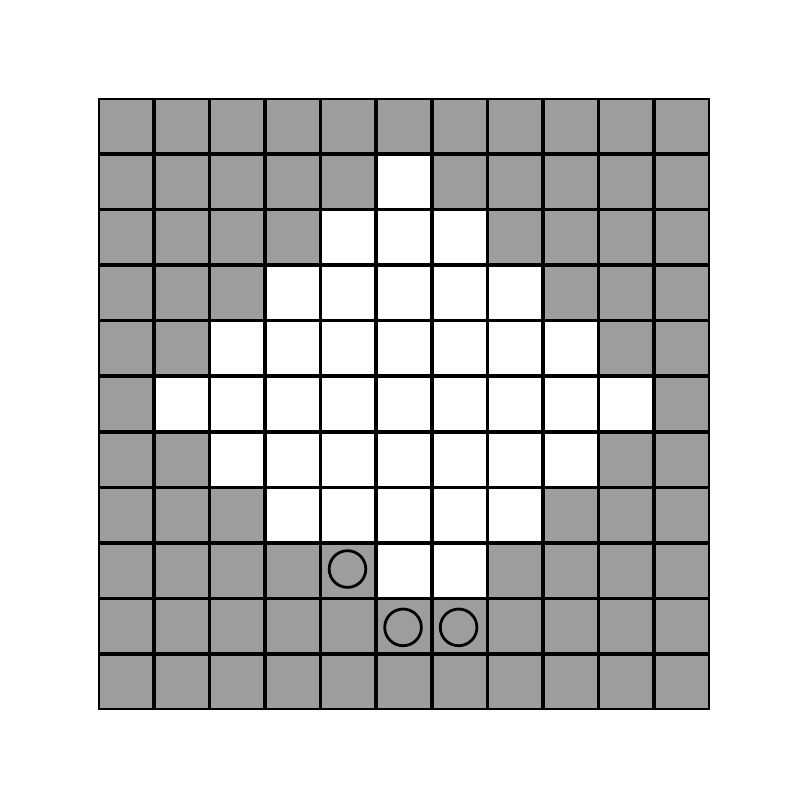}
				\caption{$\flatadjb_{9,8,3}$}
				\label{fig:4flatadjb}	
			\end{subfigure} &
			\begin{subfigure}{0.2\textwidth}\centering					
				\includegraphics[width=\linewidth]{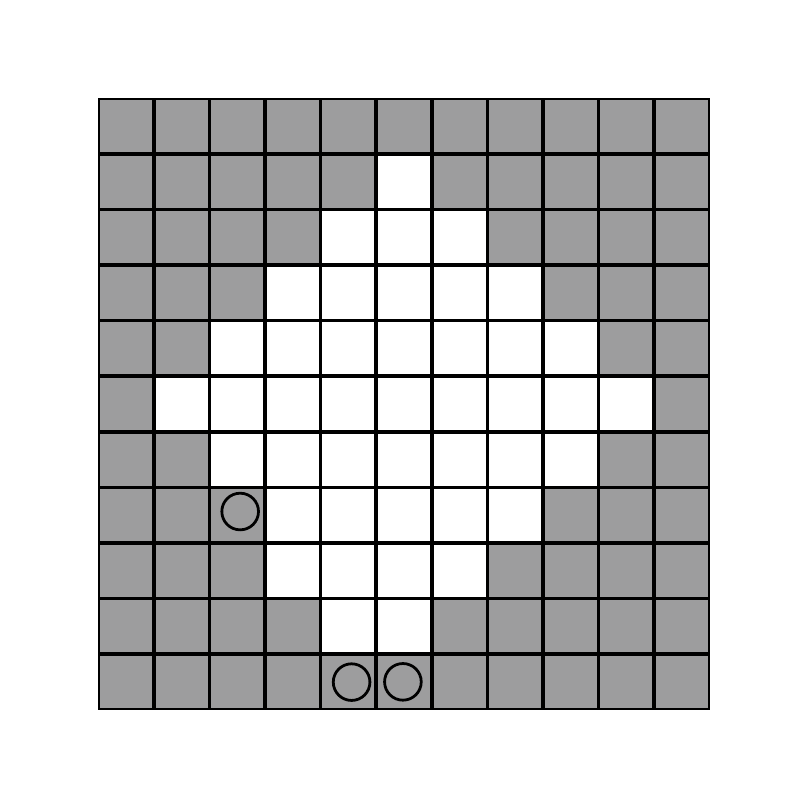}
				\caption{$\flatadja_{9,9,2}$}
				\label{fig:4flatadja}	
			\end{subfigure}	\\
		\end{tabular}
		\caption{Transit Clusters}
		\label{fig:growClusterB}
	\end{figure}
\end{center}

Let $\mathcal{D}$ be the set of configurations listed above together with their reflections and rotations in the plane across all possible values of $w$, $h$ and $\ell$.

%\begin{align*}
% \mathcal{D} & = \{\stable_{w,h} | w \geq 3, h\geq 3\} \cup \{\flatfour_{w,h} | w \geq 2, h\geq 2\} \cup \{\flatopp_{w,h} | w\geq 4, h\geq 3\} \\ & \cup \{\flatadj_{w,h}| w\geq 5, h\geq5\}
%\cup \{\flatoppb_{w,h,\ell} | w\geq 4, h\geq 3, \ell \geq 1\} \cup \{\flatadjb_{w,h,\ell} | w\geq 5, h \geq 5, \ell \geq 1\}\\  &  \cup \{\flatoppa_{w,h,\ell}| w\geq 4, h\geq 4, \ell\geq 1\}  \cup  \{\flatadja_{w,h,\ell}| w\geq 5, h \geq 5\} 
%\end{align*}

To show how transitions occur between elements of $\mathcal{D}$, consider the sequences of configurations given in Figure \ref{fig:4example}. Recall that elements of $\mathcal{D}$ are considered equivalent up to reflection and rotation. Examining $C_0$ we see that there are $4!$ updating permutations. Let $u_1$ and $u_2$ be the pair of horizontal weak vertices so that $u_1$ is to the left of $u_2$, and $v_1$ and $v_2$ be the pair of vertical weak vertices so that $v_1$ is above $v_2$. In considering the possible updating permutations we note that for each pair of adjacent weak vertices it only matters which of the pair comes first. That is, the sequence $u_1,v_1,u_2,v_2$ will give the same resulting configuration as $u_1,u_2,v_1,v_2$. By also considering the symmetries of the cluster,  the $4!$ possible updating sequences may be partitioned in the three equivalence classes. Sequences $C$,$D$ and $E$ give the resulting sequence of configurations given by the equivalence classes with representative elements (permutations): $(u_1,v_1,u_2,v_2), (u_2,v_1,u_2,v_1)$  and $(u_2,v_2,u_1,v_1)$, respectively.  By analysing the equivalence classes, $C$ and $D$ each occur with probability $\frac{1}{4}$, and $E$ occurs with probability $\frac{1}{2}$. After proceeding through the forced iterations, we see that in each case we arrive at an element of $\mathcal{D}$. In particular, $\flatadj_{8,7}$ transitions to $\stable$ with probability $\frac{1}{4}$, to $\flatadj_{8,9}$ with probability $\frac{1}{4}$ and to $\flatadja_{9,8,3}$ with probability $\frac{1}{2}$. We note that by a similar analysis for other elements of $\mathcal{D}$, if $C_i^\prime \in \mathcal{D}$, then $C_{i+1}^\prime \in \mathcal{D}$.

\begin{figure}
	\begin{center}
		\includegraphics[scale = 0.25]{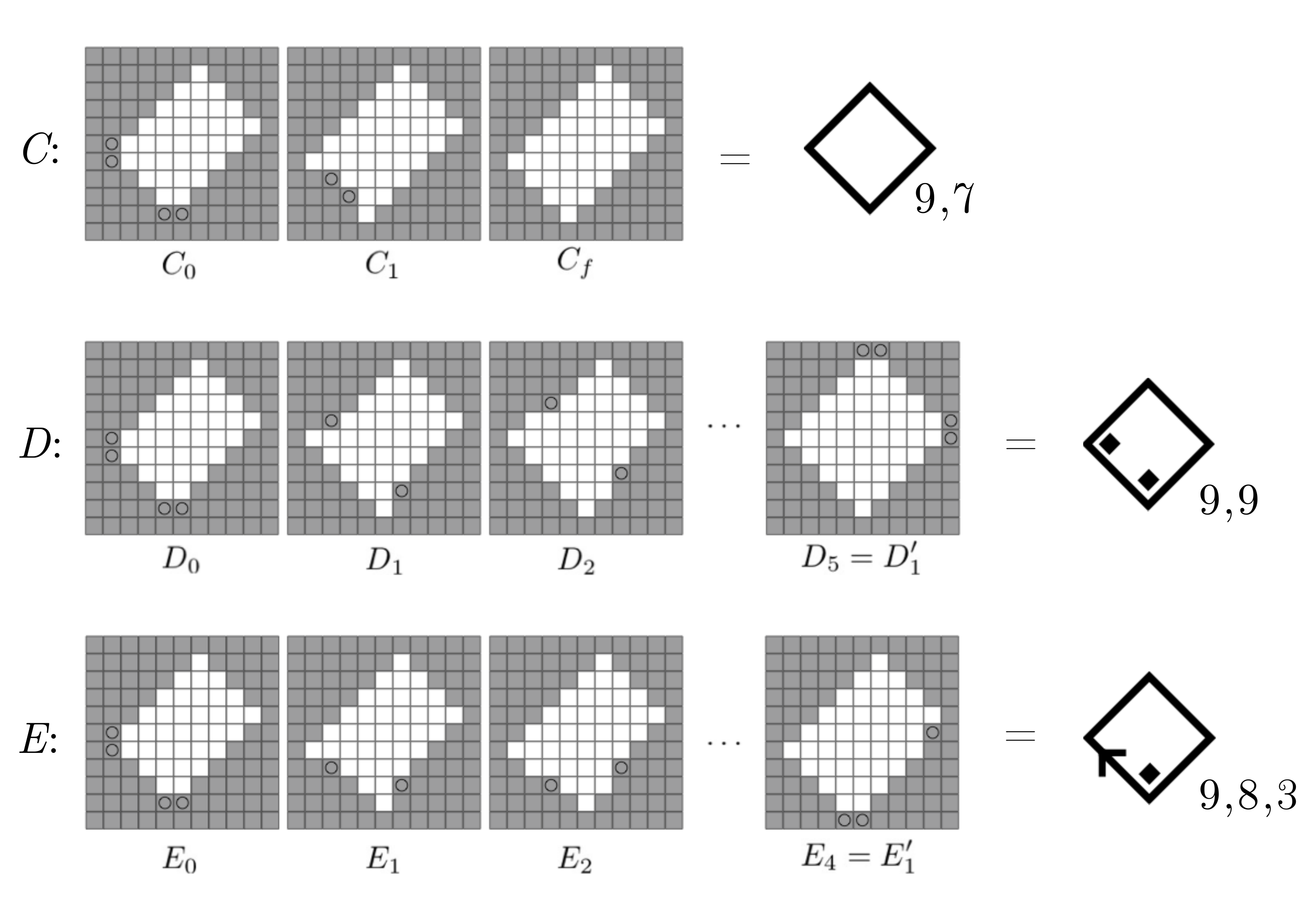}
		\caption{Evolution of $\flatadj_{8,7}$}
		\label{fig:4example}
	\end{center}
\end{figure}

Figure \ref{fig:transits} gives a partial transition diagram for the elements of $\mathcal{D}$. The transitions shown do not depend on the particular values of $w$ and $h$. For example, regardless of the values of $w$ and $h$, an element of $\mathcal{D}$ of the form $\flatfour_{w,h}$ will transition to $\stable_{w+2,h+2}$ with probability $\frac{1}{8}$. However the transition of an element of the form $\flatadj_{w,h}$ will only transform to an element of the form $\flatadj_{w+1,h+1}$ (with probability $\frac{1}{4}$) if $w = h$. From this diagram we see that each element of $\mathcal{D}$ that is not a stable cluster can transition to a stable cluster after at most $3$ basic time-steps. Such a transition will occur with probability at least $\frac{1}{8}$. This fact is given by the following lemma.
 
\begin{figure}
	\begin{center}
		\includegraphics[scale = 0.5]{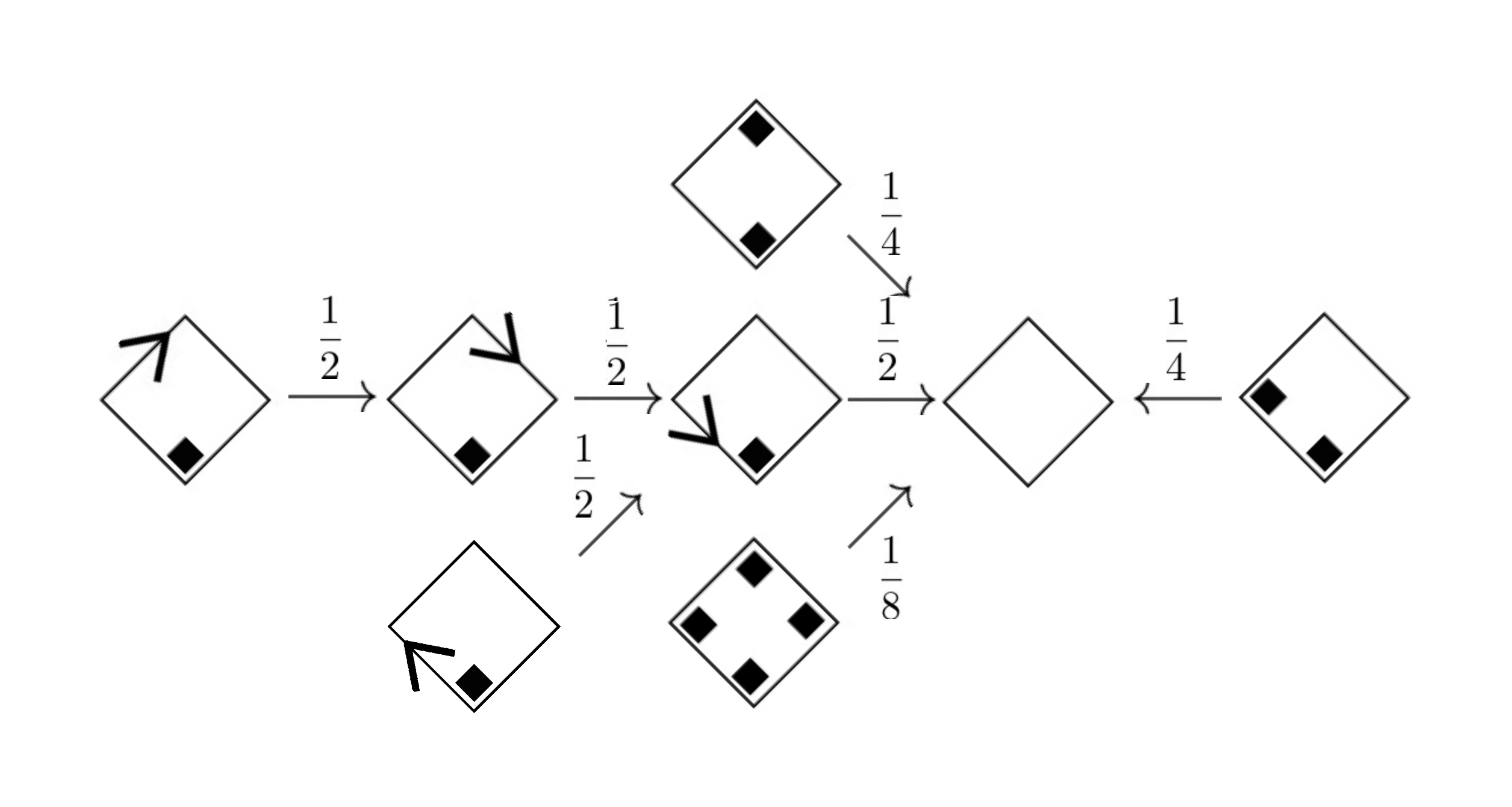}
		\caption{Transitions between elements of $\mathcal{D}$ that do not depend on $h$ and $w$}
		\label{fig:transits}
	\end{center}
\end{figure}

\begin{lemma}\label{lem:stableTime}
If $C_i^\prime \in \mathcal{D}$, then with probability at least $\frac{1}{8}$ either $C_{i+1}^\prime = \stable$ or $C_{i+2}^\prime = \stable$ or $C_{i+3}^\prime = \stable$.
\end{lemma}

Figure \ref{fig:transits} is constructed to show that regardless of the configuration in $C_i^\prime$, it is possible that one of $C_{i+1}^\prime, C_{i+2}^\prime$ or $C_{i+3}^\prime$ is $\stable$. As such, there are transitions between elements of $\mathcal{D}$ that can be deduced, but are not shown in Figure \ref{fig:transits}. For example if $w=h$ then $\flatadj_{w,h}$ transitions to $\flatadj_{w+2,h+2}$ with probability $\frac{3}{4}$. This transition, and many others, are omitted from Figure \ref{fig:transits} as they are not relevant to Lemma \ref{lem:stableTime}. As the omitted transitions are both cumbersome to enumerate and do not aid in the analysis below, we shall refrain from enumerating them.

Lemma \ref{lem:stableTime} implies the following corollary about the termination of the PD process seeded with a $4$-cluster of cooperators in an infinite grid of defectors.

\begin{corollary}
	If $C_0$ is a $4$-cluster of cooperators in an infinite grid of defectors, then the PD process terminates with probability $1$.
\end{corollary}

\begin{proof}
By Lemma \ref{lem:stableTime}, if $C_i^\prime \neq \stable$, then with probability at least $\frac{1}{8}$ either $C_{i+1}^\prime = \stable$ or $C_{i+2}^\prime = \stable$ or $C_{i+3}^\prime = \stable$. Therefore the probability that the process has not stabilised before basic time step $3i$ is bounded above by $\frac{7}{8}^i \to 0$ as $i \to \infty$.
\end{proof}

\begin{corollary}\label{cor:4clusterbound}
	The growth of a $4$-cluster of cooperators in an infinite field of defectors is contained within a ball of radius $2i +4$ with probability at least $1-\frac{7}{8}^{\lfloor i/3\rfloor}$.
\end{corollary}

\begin{proof}
	Consider the basic sequence $C^\prime_0, C^\prime_1, \dots$, where $C^\prime_0$ is a $4$-cluster of cooperators in an infinite field of defectors. The probability that the process survives until  time step $3i$ (with respect to the basic sequence) is bounded above by $\frac{7}{8}^i$ as if process survives to reach step $k$ then it survives to reach step $k+3$ with probability at least $\frac{7}{8}$.
	% Between $C^\prime_i$ and $C^\prime_{i+1}$ the width parameter increases by at most $2$ and the height parameter increases by at most $2$.  By Figure \ref{fig:4clusters}, the $4$-cluster in $C^\prime_1$ is contained by a ball of radius $4$.  Therefore the process survives until basic time-step $i$, then it is contained in a ball of radius $2i + 4$. 
	Similarly, if the growth is contained in a ball of radius $2i + 4$, then the process must have terminated before round $i$. Therefore the probability  that the growth is contained within a ball of radius $2i+4$ is bounded below by the probability that the process terminates no later than basic time step $3i$, which is bounded below by $1-\frac{7}{8}^{\lfloor i/3\rfloor}$.
	
\end{proof}

Corollary \ref{cor:4clusterbound} allows us to provide a crude upper bound on the expected growth of a $4$-cluster of cooperators in an infinite field of defectors.

\begin{theorem}\label{thm:clusterbound}
	If $C_0$ is a $4$-cluster of cooperators in an infinite field of defectors, then the expected number of cooperators in a stable configuration is bounded above by $$\sum_{j>0}   \left(\frac{7}{8}\right)^j\frac{1}{8}\left( (6j+8)^2 + (6j+7)^2  \right) < 8919.$$
	\end{theorem}

\begin{proof}
	The probability that the process ends at $C^\prime_{3j}, C^\prime_{3j+1}$ or $C^\prime_{3j+2}$ is bounded below by $\frac{7}{8}^j\frac{1}{8}$. In round $C^\prime_{3j+2}$ the growth is contained in a ball of radius $2(3j+2) + 4 = 6j+8$. The ball of radius $6j+8$ in the grid contains $(6j+8)^2 + (6j+7)^2$
	  vertices.
\end{proof}

We note that, though the sum in Theorem \ref{thm:clusterbound} does converge, the value is much more than the actual  growth of a $4$-cluster as observed in simulations. This bound may be improved by observing that in the first few configurations $\flatadj$ appears very frequently. When $C^\prime_i = \flatadj$ the process terminates at $C^\prime_{i+1}$ with probability $\frac{1}{2}$. Further, the assumption that between $C^\prime_i$ and $C^\prime_{i+1}$ each of the height and width grow by $2$ is not true, unless $C^\prime_i = \flatfour$, otherwise the sum of the height and width  grows by at most $2$.

\section{Growth in the Toroidal Grid} \label{sec:Growth}
Recall the following important tools of probability theory~\cite{A11}.

 \emph{Markov's Inequality:} \emph{If $X$ is a random variable such that $X > 0$, then $Pr(X \geq a) \leq \frac{ \mathbb{E}(X)}{a}$.}

 \emph{Chebyshev's Inequality:} \emph{If $X$ is a random variable such that $X > 0$, then $Pr(|X- \mathbb{E}(X)|) > a) \leq \frac{Var(X)} {a^2}$}. 

Additionally, consider the following definitions from the study of asymptotic analysis

An event  $E$ in some probability space parametrised by $n$ holds \emph{asymptotically almost surely (a.a.s.)} if $Pr(E = 1) \to 1 \mbox { as } n \to \infty$.

Let $f,g:\mathbb{Z} \to \mathbb{R}$. We say that $f$ is \emph{much smaller than} $g$, denoted $f \ll g$,  if $\frac{f(n)}{g(n)} \to 0$ as $n\to \infty$. We say that $f$ and $g$ are \emph{asymptotically equal}, denoted $f \sim g$, if $\frac{f(n)}{g(n)} \to 1$ as $n\to \infty$.

We begin by considering the simplest non-trivial toroidal grid -- the $n$-vertex cycle:

 \begin{lemma}\label{lem:cycleNoWeak}
 	Let $C_t$ be a configuration on the $n$-vertex cycle ($n\geq 3$). The follow statements hold.
 	\begin{enumerate}
 		\item $C_t$ contains no weak defectors.
 		\item A weak cooperator in $C_t$ is either isolated, or has a neighbour that is either an isolated defector, or has a cooperator neighbour with a defector neighbour.
 		\item $C_{t+1}$ contains no isolated defectors.
 	\end{enumerate}
 \end{lemma}
 
 \begin{proof}
  Statement $1$ follows from observing that a weak cooperator must have a defector neighbour. However, such a neighbour has score at most $1$, which is strictly less than a defector with a cooperator neighbour.
  
 Let $v$ be a weak cooperator that is neither isolated nor has a neighbour that is an isolated defector. Since $v$ is weak it must have at least one defector neighbour. Since this defector is not isolated, then it has score $1 + \epsilon$. Since $v$ is weak but not isolated, it must be that a cooperator neighbour scores at most $1$. This implies that such a neighbour has a defector neighbour.
 
 Statement $3$ follows from Statement $1$ and by observing that both neighbours of an isolated defector are weak and so at least one of them will change to a defector by the end of round $t$.
 \end{proof}
 
 \begin{lemma}\label{lem:CyclePath}
 	If $P$ is a maximal induced path of cooperators in $C_0$ on $4 < t < n-2$ vertices, then the interior vertices of $P$ are initial persistent cooperators.
 \end{lemma}
 
 \begin{proof}
 	The interior vertices of $P$ are strong in $C_0$. An end vertex of a maximal induced $t > 2$ vertex path of cooperators is weak if and only if it has a neighbour that is an isolated defector. Since there are no isolated defectors in $C_i$ for all $i >0$, any induced path of at least $3$ cooperators will be strong in $C_1$. 
 \end{proof}
 
 \begin{theorem} \label{thm:Cycleend}
 	If $C_0$ is a configuration on the $n$-vertex cycle then there exists $i$ such that $W_i = \emptyset$.
 \end{theorem}
 
 \begin{proof} 
 It follows from Lemma \ref{lem:cycleNoWeak} that the number of defectors increases monotonically at the end of each round. Since the number of defectors is bounded above by $n$, it follows that there exists $i \geq 0$ such that $W_i = \emptyset$.
 \end{proof}
 
\begin{corollary}
 	The Prisoner's Dilemma process terminates for all initial configurations of the $n$-vertex cycle.
 \end{corollary}
 
 Since the Prisoner's Dilemma process terminates for all initial configurations of the $n$-vertex cycle, we may consider $r_f$, the density of cooperators in the stable configuration. When we consider $C_0$ formed by assigning each vertex to independently be a cooperator with fixed probability $p \in (0,1)$ we find upper and lower bounds for expected value of $r_f$ as a function of $p$.
 
By employing methods similar to previous work in this area,  we arrive at the following result.
 
\begin{theorem}[Guzm\'an Pro, Janssen \cite{SJ15}]\label{thm:CycleMain}
	Consider the PD process on the  $n$-vertex cycle  where $C_0(v) = 1$ with  probability $p\in (0,1)$ for all $v \in V(G)$.
	$$3p^5 - 2p^6 + o(1) < \mathbb{E}(r_f) < p - p(1-p)^2 - 2p^2(1-p)^2 + o(1).$$
\end{theorem}
 
 \begin{proof}
  Begin by observing that a configuration of the cycle in which each vertex is a cooperator is stable, and a configuration in which a single vertex is a defector becomes stable when exactly one of the neighbours of the defector changes to be a defector.
  Consider a path, $P$, on $k+2 < n$ vertices in the $n$-vertex cycle where the ends of the path are  defectors and the interior vertices are all cooperators. The expected number of such paths in $C_0$ is given by $np^k(1-p)^2$. If $k < 3$ then every vertex of $P$ will be a defector in $C_f$. If $k  > 4$, then at most $k$ vertices of $P$ will be cooperators in $C_f$. If $k = 3$ then at most three vertices of $P$ will be cooperators in $C_f$. If $k = 4$ then at most $4$ vertices of $P$ will be cooperators in $C_f$.
 	Therefore
 	
 	\begin{align*}
 		\mathbb{E}(r_f) &< \frac{n^2p^n + n(n-2)p^{n-1}(1-p) + \sum^{n-2}_{k = 3} nkp^k(1-p)^2}{n}\\ &<  np^n + (n-2)p^{n-1}(1-p) + \sum_{k = 3}^{\infty} kp^k(1-p)^2\\ 
 		\end{align*}
 		
 	By the method of generating functions we find that $\sum_{k = 3}^{\infty} kp^k(1-p)^2 =p - p(1-p)^2 - 2p^2(1-p)^2$. Further we observe that $np^n + (n-1)p^{n-1}(1-p) \to 0$ as $n \to \infty$. Thus we conclude.
 		
 			\begin{align*}
 			\mathbb{E}(r_f) <  p - p(1-p)^2 - 2p^2(1-p)^2 + o(1).
 			\end{align*}
 	
 To find a lower bound, notice that if $k  > 4$, then at least $k-2$ vertices of $P$ will be cooperators in $C_f$. Therefore
 	\begin{align*}
 			\mathbb{E}(r_f) &> \frac{n^2p^n + n(n-2)p^{n-1}(1-p) + \sum^{n-1}_{k = 5} n(k-2)p^k(1-p)^2}{n}\\ 
 			& =  np^n + (n-2)p^{n-1}(1-p) + \sum^{n-1}_{k = 5} (k-2)p^k(1-p)^2\\
 	\end{align*}
 	
 	Notice that $ np^n + (n-2)p^{n-1}(1-p) \to 0$ as $n \to \infty$. Since $\sum^{n-1}_{k = 5} (k-2)p^k(1-p)^2{n} > 3p^5 - 2p^6$, we conclude that  $$3p^5 - 2p^6 +o(1) <  \mathbb{E}(r_f) \leq p - p(1-p)^2 - 2p^2(1-p)^2 + o(1).$$

 \end{proof}
 
 \begin{theorem}\label{thm:cycleAAS}
 Consider the PD process on the $n$-vertex cycle  where $C_0(v) = 1$ with fixed probability $p\in (0,1)$ for all $v \in V(G)$.
 		$$3p^5(1-p)^2 \leq r_f \leq p$$ asymptotically almost surely.
 \end{theorem}
 
 \begin{proof}
 Let $C_n = v_1, v_2 ,\dots, v_n$ be the cycle on $n$ vertices. Let $X_i$ be the indicator variable that is $1$ if $v_i$ becomes a persistent collaborator at some point during the process. By Lemma \ref{lem:CyclePath}, a path of $5$ cooperators in $C_0$ will have its internal vertices be cooperators for all subsequent rounds. Thus, $v_i$ will be an initial persistent collaborator if it is an internal vertex of a maximal induced path of cooperators of length $5$. Therefore $Pr(X_i = 1) > 3p^5(1-p)^2$. We note that this inequality is strict as there are other ways for a vertex to be a persistent collaborator other than being an internal vertex of a maximal induced path cooperators of length $5$. For example, such a vertex could be an internal vertex of a maximal induced path of cooperators of length $3$ that is surrounded by non-isolated collaborators.
 
  Let $X = \sum X_i$. By linearity of expectation $\mathbb{E}(X) > 3np^5(1-p)^2$. Let $p^\prime = Pr(X_i = 1)$ and let $\epsilon = p^\prime -  3p^5(1-p)^2 > 0$. Note that   $\sum_i Var(X_i) = np^\prime(1-p^\prime)$.
 	By Chebyshev's inequality
 	\begin{align*}
 	Pr(X \leq 3np^5(1-p)^2) &\leq Pr(| X - \mathbb{E}(X) | \geq \epsilon n )\\
 	& =  \frac{\sum Var(X_i) + \sum_{i \neq j} Cov(X_i,X_j)}   {{(\epsilon n)^2}}\\
 	& = \left( np^\prime(1-p^\prime) + \sum_{i\neq j} Cov(X_i,X_j)\right) \cdot\frac{1}{(\epsilon n)^2}.
 \end{align*}
 
 Observe that if $d(v_i,v_j) \geq 7$ then $Cov(X_i,X_j) = 0$. Therefore  $\sum_{i\neq j} Cov(X_i,X_j)$ has at most $14n$ non-zero terms. Since each of these terms is bounded above by $1$, we conclude  $\sum_{i\neq j} Cov(X_i,X_j) \leq 14n$. 
 %As $p$ is fixed, there exists a constant $c^\prime$ such that $c^\prime n$ = $\sum_{i \neq j} Cov(X_i,X_j)$.  
 Therefore
 
\begin{align*}
Pr(| X - \mathbb{E}(X) | \geq \epsilon n ) &\leq \frac{np^\prime(1-p^\prime) + 14n}{\epsilon^2 n^2}\\ 
&= \frac{p^\prime(1-p^\prime) + 14}{\epsilon^2n}.
\end{align*}

Since $p$ is constant, this expression goes to $0$ as $n \to \infty$. Therefore,  a.a.s., $X = \mathbb{E}(X)(1+o(1))$. This implies, a.a.s., $r_f \geq 3p^5(1-p)^2 $.
 
Observe that by Lemma \ref{lem:cycleNoWeak} each defector in $C_0$ is an initial persistent  defector.  Let $Y_i$ be the indicator variable that is $1$ is $v_i$ is a persistent defector at some point during the process. Since $v_i$ will be an initial persistent defector if $v_i$ is a defector, then $Pr(Y_i = 1) > (1-p)$. As before, we note that this inequality is strict as there are other ways for a vertex to become a persistent defector. For example, any cooperator that transitions to become a defector will be a persistent defector.
 Let $q^\prime = Pr(Y_i = 1)$ and $\epsilon = q^\prime - (1-p) > 0$. We proceed as in the previous case, noting that $Cov(Y_i,Y_j) = 0$ for all $i \neq j$. Therefore, a.a.s., the final density of defectors is at least $(1-p) + o(1)$.  Therefore a.a.s., $r_f \leq p$. 
 \end{proof}

We now consider the behaviour of $r_t$ for various regimes of $p$ on the  $n\times n$ toroidal grid. In particular, we examine two cases:  $p$ as a fixed constant and  $p$ as a  function of $n$.  In the former case we follow a similar argument to that of Theorem \ref{thm:cycleAAS} to find a lower bound for $r_f$. In the latter case we consider the growth of small clusters of collaborators in an infinite field of defectors to find bounds on $r_f$ when $p(n) \to 0$ as $n \to \infty$.
	
	\begin{theorem}\label{thm:p13}
		Consider the PD process on the $n\times n$ toroidal grid  where $C_0(v) = 1$ with  probability $p \in (0,1)$ for all $v \in V(G)$. Asymptotically almost surely $r_t > p^{13}$ for all $t \geq 0$.
	\end{theorem}
	
	\begin{proof}
	%	By Corollary \ref{cor:initialPersist}, if there exists at least $m$ initial persistent cooperators, then $r_t >\frac{m}{n^2}$ for all $t \geq 0$. Let $p^\prime$ be the probability that a vertex is an initial persistent cooperator. We proceed using Chebyshev's inequality, letting $\epsilon = p^\prime - p^{13}$.

		Let $X_i$ be the indicator variable for the property that $v_i$ is a collaborator with no isolated defector at distance at most two. We note by  Corollary \ref{cor:initialPersist} that if $X_i=1$, then $v_i$ is an initial persistent cooperator. Observe that if each vertex in $N[v_i] \cup N^2[v_i]$  is  a cooperator in $C_0$, then $X_i =1$. Therefore $Pr(X_i = 1) > p^{13}$. 	We note that this inequality is strict as there are local configurations other than an entire closed second neighbourbood as cooperators that would satisfy the conditions required to have $X_i = 1$.

		Let $X = \sum  X_i$. By linearity of expectation $\mathbb{E}(X) > n^2p^{13}$.		
		Since $p$ is a fixed constant, this quantity goes to infinity as $n \to \infty$. 
		Let $p^\prime  =  Pr(X_i = 1)$. Notice that $p^\prime$ is a fixed constant in $(p^{13},1)$. Note that $\sum_i Var(X_i) = n^2p^\prime(1-p^\prime)$. Let $\epsilon = p^\prime - p^{13} > 0$.
		By Chebyshev's inequality
		\begin{align*}
	Pr(X \leq p^{13}n^2) &\leq Pr(| X - \mathbb{E}(X) | \geq \epsilon n^2 )\\
	& =  \frac{\sum Var(X_i) + \sum_{i \neq j} Cov(X_i,X_j)}   {{(\epsilon n^2)^2}}\\
	& =  \frac{Var(X) + \sum_{i \neq j} Cov(X_i,X_j)}   {{(\epsilon n^2)^2}}\\
	& = \left( n^2p^\prime(1-p^\prime) + \sum_{i\neq j} Cov(X_i,X_j)\right) \cdot\frac{1}{(\epsilon n^2)^2}.
		\end{align*}
	
	Consider a pair of vertices $v_i,v_j$ such that $d(v_i,v_j) \geq 8$. Since the value of  $X_i$ is determined by strategies of  vertices at distance no more than $4$ from $v_i$, it follows that $Cov(X_i,X_j) = 0$.
	Thus  $\sum_{i\neq j} Cov(X_i,X_j)$ has at most $113n^2$ non-zero terms, as in the toroidal grid each vertex has $113$ vertices at distance no more than $7$.
	Since each of these terms is bounded above by $1$, we conclude  $\sum_{i\neq j} Cov(X_i,X_j) \leq 113n^2$.
		
\begin{align*}	
Pr(| X - \mathbb{E}(X) | \geq \epsilon n^2 ) &\leq \frac{n^2p^\prime(1-p^\prime) + 113 n^2}{\epsilon^2 n^4}\\ 
&= \frac{p^\prime(1-p^\prime) + 113}{\epsilon^2n^2}.
\end{align*}

Since each of $p^\prime$ and $\epsilon$ are constant, this quantity goes to $0$ as $n \to \infty$. Therefore a.a.s, $r_f \geq p^{13}$. 
	\end{proof}
	
We turn now to studying the process when $p$ is taken to be a function of $n$ rather than as a fixed constant. When $p$ is taken to be a constant, any $k$-cluster of collaborators can be expected to appear given sufficiently large $n$. However, by taking $p$ as a function of $n$, we can, in a sense, control the types of $k$-clusters that are expected to appear as $n \to \infty$.
By choosing $p(n)$ sufficiently small so that $5$-clusters of cooperators are not expected to appear, we may use our observations about the growth of small clusters of cooperators to predict the final number of cooperators. Similarly, by choosing $p(n)$ sufficiently large so that $2$-clusters of defectors are not expected to appear, we may use our observations about the growth of $1$-clusters of defectors to predict the final number of defectors given a fixed value of $p$.

In studying $k$-clusters in an initial configuration, we note that any $k$-cluster $K$  may be uniquely indexed by the vertex $v \in K$ such that of the vertices in the bottom-most row of $K$, $v$ is the vertex in the left-most column. We call such a vertex the lower left corner of $K$.

We note that, though the number of cooperators in a $k$-cluster is fixed by definition (i.e., $k$), the number of defectors adjacent to a vertex of $K$ is not. The number of defectors on the perimeter of $K$ depends on the particular shape. For example, the number of defectors on the perimeter of a $3$-line of collaborators is $8$, but the number of defectors on the perimeter of a $3$-corner of collaborators is $7$. However, we note that for an $k$, the number of defectors adjacent to a vertex of $k$-cluster of collaborators is bounded by $3k$, as each collaborator of the $k$-cluster has at most $3$ defector neighbours (when $k > 1$). 

\begin{lemma}\label{lem:expectK}
Consider the PD process on an $n \times n$ toroidal grid where $C_0(v) = 1$ with probability $p= p(n)\to 0$ as $n \to \infty$. Let $K$ be a $k$-cluster of cooperators. The expected number of copies of $K$ in $C_0$ is  asymptotically equal  to $n^2p^k$.
\end{lemma}

\begin{proof}
Assume $p= p(n)\to 0$ as $n \to \infty$. If $K$ is a $k$-cluster of cooperators, then the probability that a particular vertex is the lower left corner of a copy of $K$ is $p^k(1-p)^c$, where $c$ is the number of defectors on the perimeter of $K$. Let $X_i$ be the indicator variable that is $1$ if $v_i$ is the lower left corner vertex of a copy of $K$ in $C_0$. Let $X = \sum X_i$. By linearity of expectation, $\mathbb{E}(X) = n^2p^k(1-p)^c \sim n^2p^k$.
\end{proof}

\begin{lemma}\label{lem:nokcluster}
	Consider the PD process on an $n \times n$ toroidal grid where $C_0(v) = 1$ with probability $p= p(n)$. If $p \ll n^{-\frac{2}{k}}$, then a.a.s. there are no $k$-clusters of cooperators in $C_0$. 
\end{lemma}

\begin{proof}
	 Notice that a $k$-cluster situated in a grid is a fixed polyomino with $k$ cells. Let $\alpha_k$ be the number of polyominoes of order $k$. Assume $p \ll n^{-\frac{2}{k}}$. 
	 
	 By Markov's inequality, the probability that there is at least one $k$-cluster is bounded above by ${\alpha_kp^kn^2}$. If~$p \ll n^{-\frac{2}{k}}$, then ${\alpha_kp^kn^2} \to 0$ as $n \to \infty$. Therefore, a.a.s., there are no $k$-clusters in $C_0$. 
\end{proof}

\begin{lemma} \label{lem:heavylift}
	Consider the PD process on an $n \times n$ toroidal grid where $C_0(v) = 1$ with probability $p= p(n)$.  Let $k$ be a positive integer and $K$ be a $k$-cluster of cooperators. If  $p \gg n^{-\frac{2}{k}}$  and $p(n) \to 0$ as $n \to \infty$, then the number of copies of $K$ in $C_0$ is $n^2p^{k} (1 + o(1))$ a.a.s..
\end{lemma}

\begin{proof}
	Assume that $p=p(n)$ such that $p \gg n^{-\frac{2}{k}}$  and $p(n) \to 0$ as $n \to \infty$.
	Let $X_i$ be the indicator variable that is $1$ if $v_i$ is the lower left corner of a copy of  $K$ in $C_0$. Let $X = \sum X_i$, be the number of copies  of $K$ in~$C_0$. For vertices $v_i$ and $v_j$ we examine $Cov(X_i,X_j)$. There are three cases depending on the distance between $v_i$ and $v_j$.
	
	 If $v_i$ and $v_j$ are sufficiently far apart so that there may be a copy of $K$ whose lower left vertex and is $v_i$ and one whose lower left vertex is $v_j$ such that their respective borders of width $1$ do not intersect, then $Cov(X_i,X_j) = 0$. How far $v_i$ and $v_j$ must be apart depends on $K$. Regardless, however if they are at distance at least $k+2$ then their respective borders of width $1$ do not intersect.
	 
	 If $v_i$ and $v_j$ are sufficiently far apart so that there may be a copy of $K$ whose lower left vertex is $v_i$ and one whose lower left vertex is $v_j$ such that there respective borders of width $1$ intersect, then~$E(X_iX_j)={p^{2k}(1-p)^{2c-d}}$, where $d$ is the number of common defectors in the intersecting borders of width $1$. Thus
	 
	 \begin{align*}
	 Cov(X_i,X_j) &= \mathbb{E}(X_iX_j) - \mathbb{E}(X_i)\mathbb{E}(X_j)\\
	 &= p^{2k}(1-p)^{2c-d} - \left(p^{k}(1-p)^{c}\right)\left(p^{k}(1-p)^{c}\right) \\
	 &= p^{2k}(1-p)^{2c-d} - p^{2k}(1-p)^{2c}.\\ 
	 \end{align*}

	 We note that this difference is maximised when these two copies of $K$ have borders who intersect in the maximum number of defector vertices. Let $d^\prime$ be the maximum number of perimeter defector vertices in which a pair of copies of $K$ may intersect. Therefore in such a case  
	
	 $$Cov(X_i,X_j) \leq (p^{2k}(1-p)^{2c-d^\prime}) - (p^{2k}(1-p)^{2c}) = p^{2k}(1-p)^{2c-d^\prime}(1-(1-p)^{d^\prime})$$ 
	
	 Further, note that when $K$ is fixed, for a fixed vertex $v_i$ there are a constant number, $c_K < 4k^2$ of vertices, $v_j$, such that $v_i$ and $v_j$ are sufficiently far apart so that there may be a copy of $K$ whose lower left vertex is $v_i$ and one whose lower left vertex is $v_j$ such that their respective borders of width $1$ intersect. 
	 The upper bound on $c_K$ comes by observing that for fixed $v_i$,  we have that $v_j$ must be at distance at most $2k$ from $v_i$.  
	 Since $c_K$, $c$ and $d^\prime$ are constant with respect to $k$, and since $p \gg n^{-\frac{2}{k}}$  and $p(n) \to 0$ as $n \to \infty$ we
	 observe that $$c_K n^2 p^{2k}\left[(1-p)^{2c-d^\prime}(1-(1-p)^{d^\prime})\right] \to 0 \mbox{ as } n \to \infty.$$ 
	 
	 Finally, if $v_i$ and $v_j$ are sufficiently close so that a copy of $K$ whose lower left vertex is $v_i$ will intersect with copy of $K$ or its border whose lower left vertex is $v_j$, then $X_iX_j = 0$, as both $v_i$ and $v_j$ cannot simultaneously be lower left vertices of a copy of $K$. In this case $Cov(X_i,X_j) = -p^{2k}(1-p)^{2c}$. As in the previous case we note that as $K$ is fixed then for a fixed vertex $v_i$, there are a constant number, $c_K^\prime$ of vertices, $v_j$, such that $v_i$ and $v_j$ are sufficiently close so that a copy of $K$ whose lower left vertex is $v_i$ will intersect with copy of $K$ whose lower left vertex is $v_j$. As in the previous case, we note that $c_K^\prime < 4k^2$.
	 Since $c_K^\prime$ and $c$ are constant with respect to $k$, and since $p \gg n^{-\frac{2}{k}}$  and $p(n) \to 0$ as $n \to \infty$ we 
	 observe that $$c_K^\prime n^2[-p^{2k}(1-p)^{2c} ]\to 0 \mbox { as } n \to \infty.$$ From these three cases we conclude 
	 $$ \sum_{i \neq j}Cov(X_i,X_j) \to 0 \mbox{ as }n \to \infty. $$

	Let $\epsilon = \left(\frac{1}{n^2p^k}\right)^{\frac{1}{4}}$. Observe that $\epsilon \to 0$ as $n \to \infty$, that $\epsilon^2n^2p^k \to \infty$ as $n \to \infty$, and that $$\sum Var(X_i) = n^2p^k(1-p)^c\left(1-p^k(1-p)^c\right).$$

	By Chebyshev's inequality
	
	\begin{align*}
	Pr (|X - \mathbb{E}(X)| > \epsilon \mathbb{E}(X)) & 
	 \leq \frac{\sum Var(X_i) + \sum Cov(X_i, X_j)}{(\epsilon\mathbb \mathbb{E}(X))^2}\\
	& \leq \frac{\sum Var(X_i) }{\epsilon^2n^4p^{2k}}       + \frac{    \sum Cov(X_i, X_j)     }{\epsilon^2n^4p^{2k}}\\	
	& = \frac{n^2p^k(1-p)^c\left(1-p^k(1-p)^c\right)}{\epsilon^2n^4p^{2k}} + \frac{\sum Cov(X_i, X_j)}{\epsilon^2n^4p^{2k}}\\
	 &   = \frac{(1-p)^c\left(1-p^k(1-p)^c\right)}{\epsilon^2n^2p^{k}}        + \frac{\sum Cov(X_i, X_j)}{\epsilon^2n^4p^{2k}}
	 \end{align*}

	Since $c$ is a positive constant and $p \to 0$ as $n \to \infty$, 
	$$(1-p)^c\left(1-p^k(1-p)^c\right) \leq 1$$ Thus, by the above remarks
	
 $$\frac{(1-p)^c\left(1-p^k(1-p)^c\right)}{\epsilon^2n^2p^{k}}        + \frac{\sum Cov(X_i, X_j)}{\epsilon^2n^4p^{2k}} \to 0 \mbox{ as } n \to \infty.$$
	
	Therefore, a.a.s., $X = \mathbb{E}(X)(1 + o(1))$. That is, the number of copies of $K$ in $C_0$ is $n^2p^{k} (1 + o(1))$ a.a.s..
\end{proof}

\begin{theorem}\label{thm:PDmain}
	Consider the Prisoner's Dilemma process on an $n \times n$ toroidal grid where $C_0(v) = 1$ with probability $p= f(n)$. 
	\begin{enumerate}
		\item  If $p \ll n^{-\frac{2}{3}}$, then  a.a.s. $r_f = 0$;
		\item  if $n^{-\frac{2}{3}} \ll p \ll  n^{-\frac{1}{2}}$, then a.a.s. $r_f = 20p^{3}(1+ o(1))$;
		\item  if $n^{-\frac{2}{4}} \ll p \ll  n^{-\frac{2}{5}}$, then a.a.s. $r_f \leq  20p^{3}(1+ o(1)) + 19\cdot2log^4(n)p^4(1+ o(1))$;
		\item  if  $1 - n^{-1} \ll p  \ll 1 - n^{-2}$, then a.a.s. $r_f = (2p - 1)(1 + o(1)) $; and
		\item  if  $1 - n^{-2} \ll p $,   then a.a.s. $r_f = 1$.		
	\end{enumerate}
\end{theorem} 

\begin{proof}\noindent
	Recall that a $k$-cluster situated in a grid is a fixed polyomino with $k$ cells. For $k > 0$ let $\alpha_k$ be the number of polyominoes of order $k$. From \cite{SLOANE} we get that $\alpha_3 = 19$ and $\alpha_4 = 63$. 
	\begin{enumerate}
%		Note that if $p \ll n^{-\frac{2}{k}}$ then a.a.s. there are no $k^\prime$-clusters of cooperators for any $k^\prime \geq k$.  By applying Markov's inequality we may show that any cluster of cooperators (defectors) in $C_0$ has a border of width sufficient to contain its growth. As such we may treat each cluster as if it is situated in an infinite field of defectors (cooperators). Using the results from Section \ref{sec:evolve} and Lemma \ref{lem:heavylift} we find the given bounds for $r_f$.

		\item If $p \ll n^{-\frac{2}{3}}$, then by Lemma \ref{lem:nokcluster} a.a.s.  there are no $k$-clusters of cooperators for $k > 2$ in $C_0$. Since each $1$- and $2$-cluster of cooperators evolves into an empty cluster, $C_f$ contains no cooperators.
		\item If $n^{-\frac{2}{3}} \ll p \ll  n^{-\frac{1}{2}}$ then by Lemma \ref{lem:nokcluster} a.a.s.  there are no $k$-clusters of cooperators for $k > 3$ in $C_0$.  Since each $1$- and $2$-cluster of cooperators evolves in to an empty cluster, $r_f$ is completely determined by the number of $3$-clusters of cooperators. 
		
		We show,  a.a.s., that none of the $n^2$ sub-squares of the grid with dimension $10 \times 10$ contains more than one $3$-cluster.  For a fixed $10 \times 10$ sub-square of the grid, the probability that it contains $t$ $3$-clusters is bounded above by $\alpha_3{100 \choose t}p^{3t}$. Therefore the probability that it contains two or more $3$-clusters is bounded above by 
		$$p^\prime = \sum_{i = 2}^{33} 6{100\choose i}p^{3i}.$$
		
		Let $X_i$ be the indicator variable that is $1$ if $v_i$ is the lower left vertex of a $10\times10$ sub-square that contains two or more $3$-clusters in $C_0$. Let $X = \sum X_i$. By linearity of expectation $\mathbb{E}(X) \leq  n^2p^\prime =  n^2\sum_{i = 2}^{33} 6{100\choose i}p^{3i}$. Since $n^2p^{3t} \to 0$ as $n\to \infty$ for all $t \geq 2$, we observe that $ \mathbb{E}(X) \to 0$ as $n \to \infty$. Therefore, by Markov's inequality, a.a.s., none of the $n^2$ sub-squares of the grid with dimension $10 \times 10$  contains two or more $3$-clusters. 
		
		A $3$-cluster necessarily evolves in to an empty cluster or a stable $5$-cluster  (see Section \ref{sec:Growth}). Therefore, the growth of a $3$-cluster is necessarily contained within a ball of radius two.
		Therefore, a.a.s., any existing $3$-cluster in $C_0$ will grow as if it is growing within an infinite field of defectors.
	
		By Lemma \ref{lem:heavylift} a.a.s. there are $4n^2p^3 (1+ o(1))$ $3$-corners and $2n^2p^3 (1+ o(1))$ $3$-lines in $C_0$. Each $3$-corner evolves to a stable configuration with $5$ collaborators with probability $\frac{1}{2}$ and to an empty cluster with probability $\frac{1}{2}$ (see Section \ref{sec:Growth}). Therefore the  number of collaborators in $C_f$ that arise from a $3$-corner is, a.a.s., $\frac{1}{2}\cdot5 \cdot 4n^2p^3 (1+ o(1))$. Each $3$-line evolves to a stable configuration with $5$ collaborators with probability $1$ (see Section \ref{sec:Growth}). Therefore the number of collaborators in $C_f$ that arise from a $3$-corner is, a.a.s., $5\cdot2n^2p^3 (1+ o(1))$.
		Therefore total number of cooperators in the stable configuration is a.a.s. $20n^2p^3(1 + o(1))$. Thus $r_f = 20p^{3}(1+ o(1))$.

		\item Let $n^{-\frac{2}{3}} \ll p \ll  n^{-\frac{2}{5}}$.  By the previous arguments we may conclude that a.a.s. there are no $k$-clusters for any $k > 4$. Similarly we may conclude that there are $2n^2p^3 (1+ o(1))$ $3$-lines  $4n^2p^3 (1+ o(1))$ $3$-corners, and $\alpha_4n^2p^4(1+o(1))$ $4$-clusters.

		To consider the growth of $4$-clusters our goal is to show  a.a.s. that the $4$-clusters are sufficiently spaced in $C_0$ so that they each may grow as if they are contained within an infinite field of defectors.
		We first show a.a.s. that none of the $n^2$ sub-squares of the grid with dimension $2\lfloor log^2(n)\rfloor \times2\lfloor log^2(n)\rfloor$  contains a pair of clusters of collaborators of size at least $3$.
		 This implies a.a.s. that each $4$-cluster has a border of width at least $2\lfloor log^2(n) \rfloor$. We then show  a.a.s. that the growth of each $4$-cluster is contained within a ball of radius $\lfloor log^2(n) \rfloor$. Using these two facts we are able to examine each $4$-cluster as if it is growing within an infinite field of defectors.

		For a fixed $2\lfloor log^2(n)\rfloor \times2\lfloor log^2(n)\rfloor$ sub-square of the grid, the probability that it contains at least two $3$-clusters is bounded above by

		\begin{align*}
		&\sum_{i = 2}^{\lceil log^4(n)\rceil} \alpha_3{4\lceil log^4(n)\rceil\choose i}p^{3i}  
		  \ll 6\sum_{i = 2}^{\lceil log^4(n)\rceil} \frac{\left(4^\lceil log^{4}(n)\rceil\right)^i} {n^{\frac{6i}{5} }}
		  < 6\sum_{i = 2}^{\lceil log^4(n)\rceil} \left(\frac{4\lceil log^{4}(n)\rceil} {n^{\frac{6}{5} }}\right)^i.
		\end{align*}
		
		For sufficiently large $n$,  $\frac{4\lceil log^{4}(n)\rceil} {n^{\frac{6}{5}}} < 1$.  Therefore the largest term of this sum occurs at $i = 2$. Thus
		
			\begin{align*}
			  \sum_{i = 2}^{\lceil log^4(n)\rceil} 6\left(\frac{4\lceil log^{4}(n)\rceil} {n^{\frac{6}{5} }}\right)^i
			 < 6{\lceil log^4(n)\rceil}\left(\frac{4\lceil log^{4}(n)\rceil}{n^{\frac{6}{5} }}\right)^2.
			\end{align*}
			
		Let $p^\prime = 6{\lceil log^4(n)\rceil}\left(\frac{4\lceil log^{4}(n)\rceil}{n^{\frac{6}{5} }}\right)^2$. Let $X_i$ be the indicator variable that is $1$ if $v_i$ is the lower left corner of a $2\lfloor log^2(n)\rfloor \times2\lfloor log^2(n)\rfloor$ sub square of the grid that contains at least two $3$-clusters. Let $X = \sum X_i$. By linearity of expectation $\mathbb{E}(X) < n^2p^\prime$. Observe that $n^2p^\prime \to 0$ as $n\to \infty$. Therefore, by Markov's inequality, a.a.s.,  none of the $n^2$ sub-squares of the grid with dimension $2\lfloor log^2(n)\rfloor \times2\lfloor log^2(n)\rfloor$  contains a pair of clusters of collaborators of size at least $3$.
		A similar argument shows that  none of the $n^2$ sub-squares of the grid with dimension $2\lfloor log^2(n)\rfloor \times2\lfloor log^2(n)\rfloor$  contains a pair  of $4$-clusters or a $4$-cluster and a $3$-cluster.

		 As $1$ and $2$-clusters necessarily disappear after $C_1$,  each growing  $4$-cluster has a border of width at least $2\lfloor log^2n\rfloor$ in $C_2$.
		 By Corollary \ref{cor:4clusterbound} the probability that the growth of a $4$-cluster in an infinite field of defectors is contained within in a ball of radius $log^2(n)$ is strictly bounded below by $1- \frac{7}{8}^{log^2(n)/3}$. 
		 Let $Y_i$ be the indicator variable that is $1$ if $v_i$ is the lower left corner of a $4$-cluster whose growth eventually escapes a ball of radius $log^2(n)$. 
		 Let $Y = \sum Y_i$. $Pr(Y_i = 1) < \frac{7}{8}^{log^2(n)/3}$. 
		 Recall that a.a.s., the number of $4-$clusters in $C_0$ is $\alpha_4n^2p^4(1+o(1))$.
		 Therefore  by linearity of expectation $\mathbb{E}(Y) < \left(\frac{7}{8}^{log^2(n)/3}\right) \alpha_4n^2p^4(1+o(1))$. This quantity goes to  $0$ as $n \to \infty$. Therefore by Markov's inequality, a.a.s., each $4$-cluster has a sufficiently large border to contain its growth. 
		 
		 Following the argument in $2$, in $C_f$ there are, a.a.s., $20n^2p^3(1 + o(1))$ collaborators arising as a result of $3$-clusters in $C_0$. Since the growth of each $4$-cluster is, a.a.s., contained within a ball of radius $log^2(n)$, the number of collaborators arising as a result of a single $4$-cluster is, a.a.s., bounded above by $2log^4(n)$. Therefore the number of collaborators in $C_f$ arising as a result of $4$-clusters in $C_0$ is $\alpha_4\cdot2log^4(n)p^4(1+ o(1))$. Thus, a.a.s., $r_f \leq 20p^{3}(1+ o(1)) + 19\cdot2log^4(n)p^4(1+ o(1))$.

		\item This follows similarly to 2 by letting $q = 1-p$, applying Lemmas \ref{lem:nokcluster} and \ref{lem:heavylift} with $q$ in place of $p$, and by observing that an isolated defector in an infinite field of cooperators necessarily grows to a stable 2-cluster of defectors.
		\item This follows similarly to 1 by letting $q = 1-p$ and applying Lemmas \ref{lem:nokcluster} and \ref{lem:heavylift} with $q$ in place of $p$.
	\end{enumerate}
	
\end{proof}

\subsection{Results of Simulation}\label{sec:sim}
We compare the results in Theorems \ref{thm:PDmain}, \ref{thm:p13} and \ref{thm:CycleMain} to simulated data for the process.

Simulation on the  $1000 \times 1000$ toroidal grid yields the plots in Figures \ref{fig:smallp1000}  and \ref{fig:largep1000}. 
For each $p \in \{0.001, 0.002, \dots, 0.07\}$ and $p \in \{0.920, 0.921, \dots, 0.999\}$, respectively, there were ten simulations. 
The points in the plots give the means of $r_f$ for each value of $p$. 
In each plot the curve gives the asymptotic prediction for $r_f$ given in Theorem \ref{thm:PDmain}. 
In Figure \ref{fig:smallp1000} we see that $r_f$ is modeled by a cubic function of $p$ for $p \in \{0.001, 0.002, \dots, 0.07\}$.
The deviation from the predicted values seen in figure can be explained by the asymptotic nature of the result in Theorem \ref{thm:PDmain}.
In particular, Theorem \ref{thm:PDmain} predicts $r_f$ as $n \to \infty$, and not for any particular value of $n$.
In Figure \ref{fig:largep1000} we see that $r_f$ is modeled by a linear function of $p$ for $p \in \{0.920, 0.921, \dots, 0.999\}$.

Simulation on the  $10000$-vertex cycle yields the plot in Figure \ref{fig:cyclesim}.
For each $p \in \{0.01, 0.02, \dots, 0.98, 0.99\}$ there were ten simulations.
The points  in the plot give the mean of $r_f$ for each value of $p$. 
In this plot the curve gives the asymptotic lower bound and upper bounds for $r_f$ as predicted in Theorem \ref{thm:CycleMain}.
The curves in this plot gives the theoretical asymptotic lower bound and upper bounds on $r_f$ predicted by Theorem \ref{thm:CycleMain}.

Simulation  on the $400 \times 400$  toroidal grid yields the plot in Figure \ref{fig:400figure}.
For each $p \in \{0.01, 0.02, \dots, 0.98, 0.99\}$ there were ten simulations.
The points  in the plot give the mean of $r_f$ for each value of $p$. 
The curve in this plot gives the theoretical asymptotic lower bound on $r_f$ predicted by Theorem \ref{thm:PDmain}. 
Though the result of the simulation shows that the bound given in Theorem \ref{thm:p13} is far from optimal, Theorem \ref{thm:p13} represents one of the  few rigorous asymptotic lower bounds for such a process in the literature.

%Code and datasets available at \cite{DJ16}.

\begin{figure}[ht]
	\begin{center}
		\begin{tabular}{cc}
			\begin{subfigure}{0.5\textwidth}\centering
				\includegraphics[width = 1\linewidth]{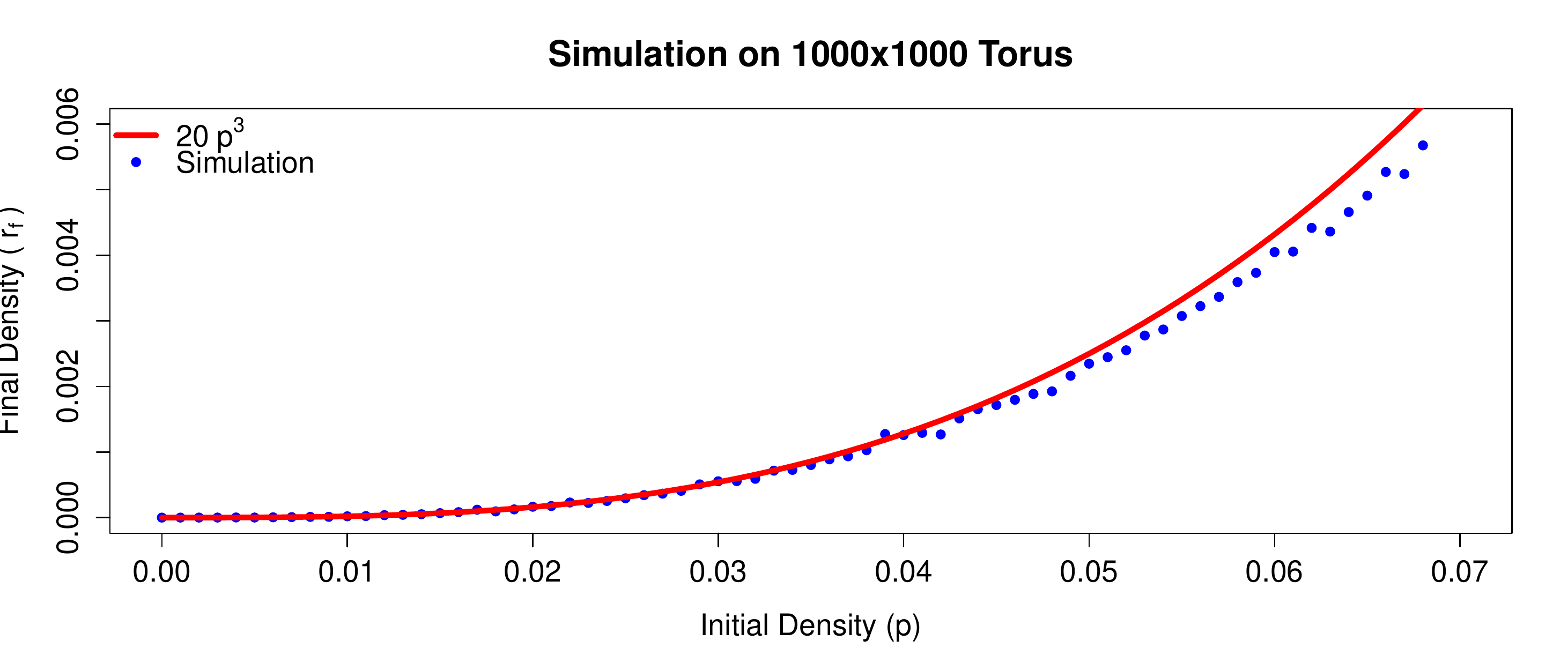}
				\caption{$p \leq 0.07$}
				\label{fig:smallp1000}
			\end{subfigure}					
			&
			\begin{subfigure}{0.5\textwidth}\centering
				\includegraphics[width = 1\linewidth]{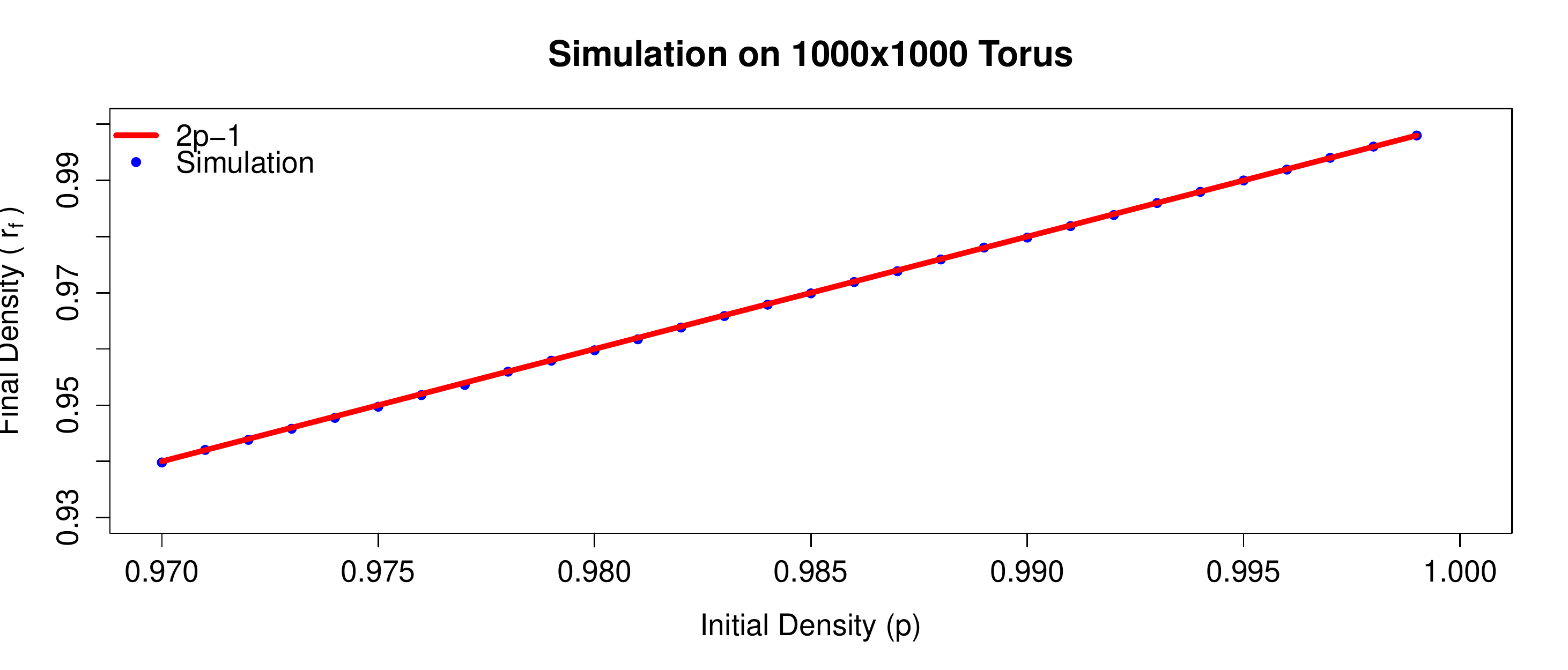}
				\caption{$p \geq 0.97$}
				\label{fig:largep1000}
			\end{subfigure}\\		
			\begin{subfigure}{0.5\textwidth}\centering\includegraphics[width=1\linewidth]{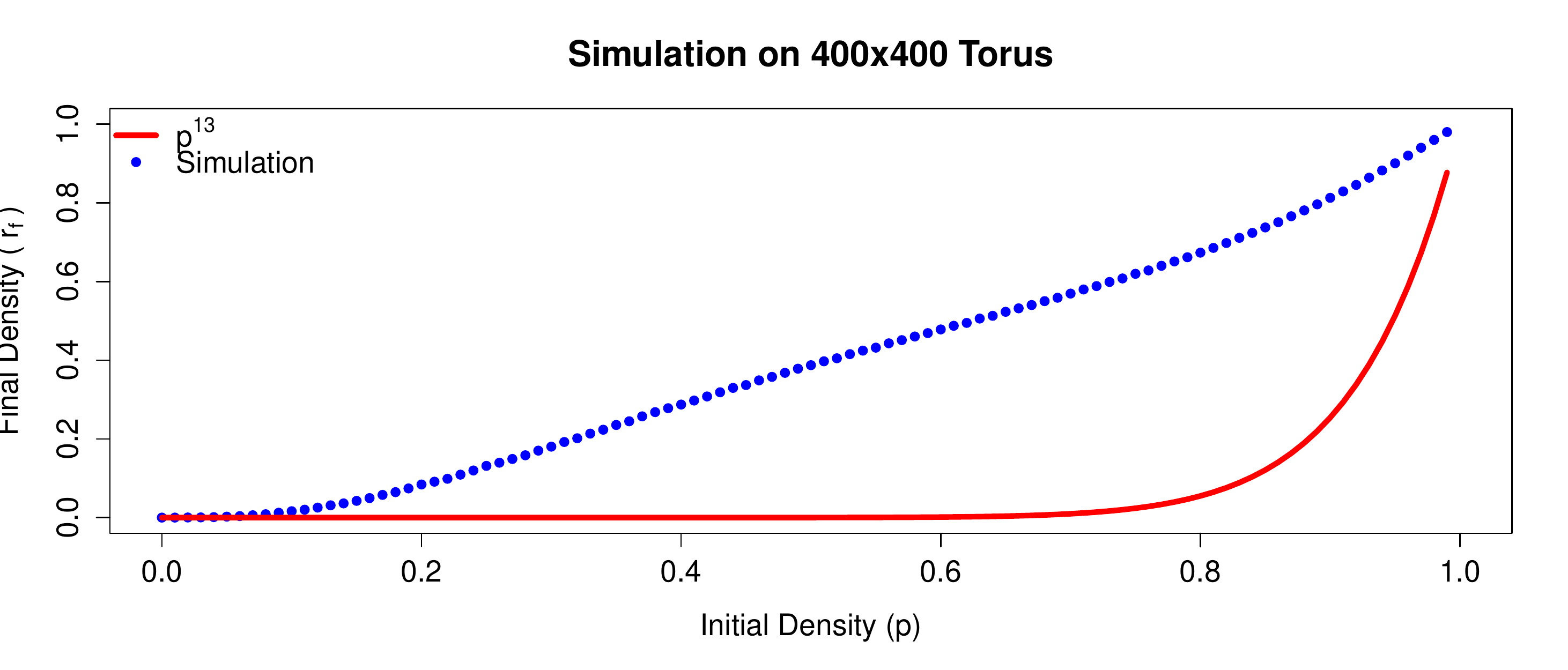}\caption{$p \in [0,1]$.} \label{fig:400figure}		\end{subfigure}			
			&
			\begin{subfigure}{0.5\textwidth}\centering
				\includegraphics[width = 1\linewidth]{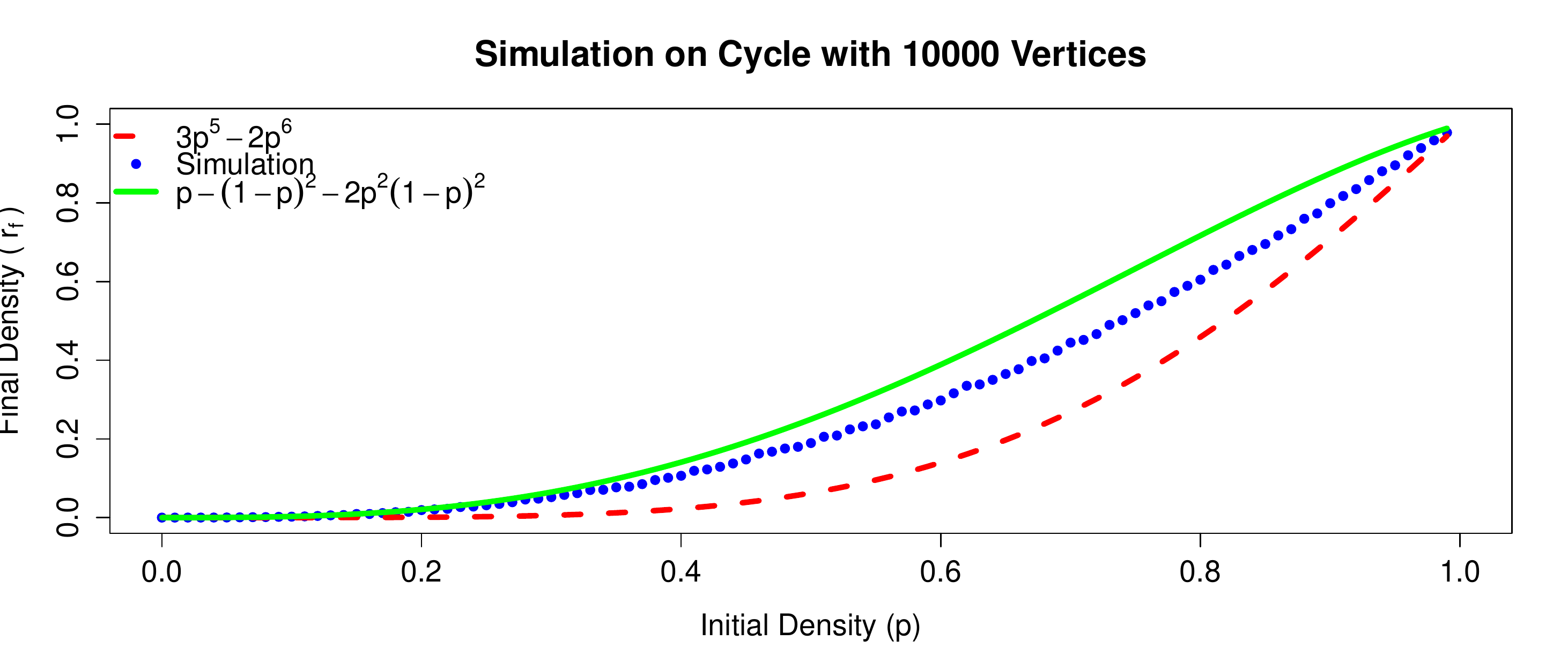}
				\caption{$p \in [0,1]$}
				\label{fig:cyclesim}	
			\end{subfigure}
		\end{tabular}
	\end{center}
	\caption{Experimental Results}
	\label{fig:sims}
\end{figure}

\section{Conclusion}
For the toroidal grid when $1 < T < \frac{4}{3}$ there is no guarantee that the process necessarily terminates. 
It is possible that the configurations exhibit periodic behaviour with non-trivial period length.
Such behaviour is observed for other values of $T$ and other updating schemes \cite{SJ15}. 
However  our simulations  suggests that the process does always terminate on the toroidal grid when $1 < T < \frac{4}{3}$. 
When the value of $p$ is constrained as in Theorem \ref{thm:PDmain}, it can be shown that the process terminates a.a.s.. 

The method used to give the bounds in Theorem \ref{thm:PDmain} certainly can be extended examine values of $p$ outside the ranges considered in Theorem \ref{thm:PDmain}. 
However, to extend the analysis to include values of $p$ for which $5$-clusters can be expected to appear would require the study of the full set of a fixed polyomino with $5$ cells,
One would need to consider how such polyominoes grow when governed by the updating process and how such the growth of such polynominoes interacts with the growth of nearby polynominoes of smaller order.
As there are twelve such polyominoes \cite{SLOANE}, this approach is infeasible in practice. 

In \cite{C05} the authors examine how asynchronous update schemes affect the evolution of one-dimensional multi-agent systems. 
Our work examines an asynchronous update scheme in a 2D-multi-agent system.
The randomized asynchronous update scheme presented herein does not have an one-dimensional analogue that is studied in \cite{C05}.
This presents two research directions following our work: (1) how does the randomized asynchronous update scheme considered in this work affect the evolution of one-dimensional multi-agent systems; and (2) how do the asynchronous update schemes examined in \cite{C05} affect the evolution of the Prisoner's Dilemma Process on toroidal grids.

	\bibliographystyle{plain}
	\bibliography{PDbib}
\end{document}